\documentclass[sigconf]{acmart}

\AtBeginDocument{%
  \providecommand\BibTeX{{%
    \normalfont B\kern-0.5em{\scshape i\kern-0.25em b}\kern-0.8em\TeX}}}

\copyrightyear{2020}
\acmYear{2020}
\setcopyright{acmcopyright}
\acmConference[CCS '20]{Proceedings of the 2020 ACM SIGSAC Conference on Computer and Communications Security}{November 9--13, 2020}{Virtual Event, USA}
\acmBooktitle{Proceedings of the 2020 ACM SIGSAC Conference on Computer and Communications Security (CCS '20), November 9--13, 2020, Virtual Event, USA}
\acmPrice{15.00}
\acmDOI{10.1145/3372297.3417254}
\acmISBN{978-1-4503-7089-9/20/11}

\usepackage{tikz}
\usepackage{appendix}
\usepackage{amsmath}
\usepackage{amsthm}
\usepackage{amssymb}
\usepackage[]{graphicx}
\usepackage{subfig}
\usepackage{algorithm}
\usepackage[noend]{algpseudocode}
\usepackage{booktabs}
\usepackage{bm}
\usepackage{multirow} 
\usepackage{threeparttable}
\usepackage{enumitem}

\newcommand{\rvb}[1]{\textcolor{black}{{}#1}}

\newcommand{\defenseName}{DualGuard}

\begin{document}
\fancyhead{}

\title{When the Differences in Frequency Domain are Compensated: Understanding and Defeating Modulated Replay Attacks on Automatic Speech Recognition}


\author{Shu Wang$^{1}$, Jiahao Cao$^{1,2}$, Xu He$^{1}$, Kun Sun$^{1}$, Qi Li$^{2}$}
\affiliation{\institution{$^{1}$Department of Information Sciences and Technology, CSIS, George Mason University}}
\affiliation{\institution{$^{2}$Institute for Network Sciences and Cyberspace, Tsinghua University; BNRist}}
\email{{swang47,xhe6,ksun3}@gmu.edu, caojh15@mails.tsinghua.edu.cn,  qli01@tsinghua.edu.cn}

\begin{abstract}

Automatic speech recognition (ASR) systems have been widely deployed in modern smart devices to provide convenient and diverse voice-controlled services. Since ASR systems are vulnerable to audio replay attacks that can spoof and mislead ASR systems, a number of defense systems have been proposed to identify replayed audio signals based on the speakers' unique acoustic features in the frequency domain. In this paper, we uncover a new type of replay attack called \emph{modulated replay attack}, which can bypass the existing frequency domain based defense systems. The basic idea is to compensate for the frequency distortion of a given electronic speaker using an inverse filter that is customized to the speaker's transform characteristics. Our experiments on real smart devices confirm the modulated replay attacks can successfully escape the existing detection mechanisms that rely on identifying suspicious features in the frequency domain. To defeat modulated replay attacks, we design and implement a countermeasure named \emph{DualGuard}. We discover and formally prove that no matter how the replay audio signals could be modulated, the replay attacks will either leave \emph{ringing artifacts} in the time domain or cause \emph{spectrum distortion} in the frequency domain. Therefore, by jointly checking suspicious features in both frequency and time domains, DualGuard~can successfully detect various replay attacks including the modulated replay attacks. We implement a prototype of DualGuard~on a popular voice interactive platform, ReSpeaker Core v2. The experimental results show DualGuard~can achieve 98\% accuracy on detecting modulated replay attacks.

\end{abstract}




\maketitle
\section{Introduction}

Automatic speech recognition (ASR) \rvb{has been a ubiquitous technique widely used in human-computer interaction systems, such as Google Assistant~\cite{GoogleAssistant}, Amazon Alexa~\cite{Alexa}, Apple Siri~\cite{Siri},  Facebook Portal~\cite{FacebookPortal}, and Microsoft Cortana~\cite{MicroCortana}. With advanced ASR techniques, these systems take voice commands as inputs and act on them to provide diverse voice-controlled services. People now can directly  use voice to unlock mobile phone~\cite{voicelock1,voicelock2}, send private messages~\cite{IEMI}, log in to mobile apps~\cite{voicelogin}, make online payments~\cite{voicepay}, activate smart home devices~\cite{voicecontrol}, and unlock a car door~\cite{voiceunlockcar}.   } 

Although ASR provides many benefits and conveniences, recent studies have found a number of attacks that can effectively spoof and mislead ASR systems~\cite{Dolphin, inaudible_0, InaudibleVC,HidCmd,CmdSong,Psychoacoustic,SkillSquatting, DangerousSkills,AfterASR,ASV2017,subbass,replayattack}. One of the most powerful and practical attacks is the audio replay attack~\cite{ASV2017,subbass,replayattack}, where a pre-recorded voice sample collected from a genuine victim is played back to spoof ASR systems. Consequently, it can easily bypass voice authentication and inject voice commands to conduct malicious activities~\cite{jain2006biometrics}. For example, a mobile device can be unlocked by simply replaying a pre-recorded voice command of its owner~\cite{ASV2017}. Even worse, the audio replay attack can be easily launched by anyone without specific knowledge in speech processing or other computer techniques. Also, the prevalence of portable recording devices, especially smartphones, makes audio replay attacks one of the most practical threats to ASR systems.

To defeat audio replay attacks, researchers have proposed a number of mechanisms to detect abnormal frequency features of audio signals, such as Linear Prediction Cepstral Coefficient (LPCC)~\cite{Lav2017}, Mel Frequency Cepstral Coefficient (MFCC)~\cite{MFCC}, Constant Q Cepstral Coefficients (CQCC)~\cite{CQCC}, and Mel Wavelet Packet Coefficients (MWPC)~\cite{LPCC}. A recent study~\cite{freq_solution} shows that the amplitude-frequency characteristics in a high-frequency sub-band will change significantly under the replay attack, and thus they can be leveraged to detect the attack. Another study~\cite{subbass} discovers that the signal energy in the low-frequency sub-bands can also be leveraged to distinguish if the voice comes from a human or an electronic speaker. 
Moreover, due to the degraded amplitude components caused by the replay noise, the frequency modulation features~\cite{FM1, FM2, ModDynamic} can be leveraged into detection. \rvb{Overall, existing countermeasures are effective on detecting all known replay attacks by checking suspicious features in the frequency domain.}


In this paper, we present a new replay attack named \emph{modulated replay attack}, which can generate replay audios with \rvb{almost} the same frequency spectrum as human voices to bypass the existing countermeasures. Inspired by the loudspeaker equalization techniques in auditory research that targets at improving the sound effect of an audio system~\cite{Equalization0}, the core idea of modulated replay attack is to compensate for the differences in the frequency domain between replay audios and human voices. Through a measurement study on ASR systems, we find the differences in the frequency domain are caused by the playback electronic speakers, which typically have a non-flat frequency response with non-regular oscillations in the passband. 
In reality, an speaker can hardly output all frequencies with equal power due to its mechanical design and the crossover nature if the speaker possesses more than one driver~\cite{Equalization1}. \rvb{Thus,} when the genuine human audio is replayed, \rvb{electronic speakers exert} different spectral gains on the frequency spectrum of the replay audio, leading to \rvb{different distortion degrees}. Typically, electronic speakers suppress the low-frequency components and enhance the high-frequency components \rvb{of the genuine human audio}.

By evaluating the transfer characteristic of electronic speakers, we are able to  customize a pre-processing inverse filter for any given speaker. By applying the inverse filter before replaying the human audio, the spectral effects caused by the speaker devices can be offset. Consequently, the attacker can produce spoofed audios that are difficult to be distinguished from real human voices in the frequency domain. We conduct experiments to demonstrate the feasibility and effectiveness of the modulated replay attack against 8 existing replay detection mechanisms using 6 real speaker devices. The experimental results show that the detection accuracy of most frequency-based countermeasure significantly drops from above 90\% to around 10\% under our attack, and even the best countermeasure using MWPC~\cite{LPCC} drops from above 97\% to around 50\%. 
One major reason is that modulated replay attack is a new type of attack that leverages loudspeaker frequency response compensation. 



To \rvb{defeat} the modulated replay attack \rvb{as well as classical replay attacks}, we propose a new dual-domain defense method \rvb{named \emph{\defenseName}} that \rvb{cross-checks suspicious features in both time domain and frequency domain, which is another major  contribution in this paper. The key insight of our defense is that it is inevitable for any replay attacks to either leave \emph{ringing artifacts}~\cite{Ringing} in the time domain or cause \emph{spectrum distortion} in the frequency domain, even if the replay audio signals have been modulated. We formally prove the correctness and universality of our key insight. In the time domain,} 
ringing artifacts will cause spurious oscillations, which  \rvb{generate a large} number of local \rvb{extreme points in replay audio waveforms. \defenseName~extracts and leverages those local extrema patterns to train a Support Vector Machine (SVM) classifier that distinguishes modulated replay attacks from human voices. In the frequency domain, spectrum distortion will generate dramatically different power spectrum distributions compared to human voices. Also, \defenseName~ applies the area under the CDF curve (AUC) of power spectrum distributions to filter out classical replay attacks.} Therefore, \defenseName~can effectively identify replay audio by performing the checks in two domains.

\rvb{We implement a prototype of \defenseName~on a voice interactive platform, ReSpeaker Core v2~\cite{ReSpeaker}. We conduct extensive experiments to evaluate its effectiveness and performance on detecting replay attacks. The experimental results show that \defenseName~can achieve about 98\% detection accuracy against the modulated replay attack and over 90\% detection accuracy against classical replay attacks. 
Moreover, we show that \defenseName~works well under different noisy environments. Particularly, the detection accuracy only decreases by 3.2\% on average even with a bad signal-to-noise ratio (SNR) of 40 dB. \defenseName~is lightweight and can be deployed to work online in real ASR systems. For example, our testbed platform takes 5.5 $ms$ on average to process a signal segment of 32 $ms$ length using 24.2\% CPU and 12.05 MB memory.} 


\rvb{In summary, our paper makes the following contributions:}
\vspace{-0.05in}

\begin{itemize} [leftmargin=*]
    \item We propose a new modulated replay attack against ASR systems, utilizing a specific software-based inverse filter to offset suspicious features in the frequency domain. By compensating the electronic speaker's non-flat frequency response in the passband, modulated replay attacks can bypass existing replay detection mechanisms.

    
    \item We design a novel defense system named \defenseName~to detect all replay attacks including the modulated replay attacks by checking suspicious features in both frequency domain and time domain. We formally prove that replay attacks cannot escape from being detected in both time and frequency domains.

    
    \item \rvb{We verify the feasibility and effectiveness of the modulated replay attack through real experiments using multiple speaker devices over existing replay detection mechanisms. We also implement a prototype of \defenseName~on a popular voice platform and demonstrate its effectiveness and efficiency in detecting all replay attacks.}
    
\end{itemize}


\vspace{-0.1in}
\section{Background}\label{sec:back}
In this section, we introduce necessary background  information on audio signal processing, ASR systems, and replay attacks.

\vspace{-0.08in}
\subsection{Audio Signal Processing}
\noindent\rvb{As there are so many technical terms on voice signal processing, we only briefly introduce two necessary terms that are tightly related to our work.}

\noindent {\bf Signal Frequency Spectrum.}
\rvb{Generally}, a signal is represented as a time-domain form $x(t)$, recording the signal amplitude at each time point. Frequency spectrum is another signal representation, providing a way to analyze the signal in the frequency domain. \rvb{ Fourier analysis~\cite{Fourierseries} can decompose a time-domain signal as} the sum of multiple sinusoidal signals of different frequencies, i.e., $x(t) = \sum_{n} A_{n} \cdot \sin{(2 \pi f_{n} t + \phi_{n})}$. The $n$-th sinusoidal signal is called the frequency component with a frequency value of $f_{n}$. The set of $\{A_{n}\}$ is called the amplitude spectrum \rvb{that} represents the amplitude of each frequency component. $\{ \phi_{n} \}$ is the phase spectrum recording the phase of each component. The frequency spectrum of a signal is the combination of amplitude and phase spectrum. 

\noindent {\bf Frequency Response.}
\rvb{Frequency response represents the output frequency and phase spectrum of a system or a device in response to a stimulus signal~\cite{Frequencyresponse}.}
\rvb{When a stimulus signal that is typically a single-frequency sine wave passes through a system, the ratio of the output to input amplitude (i.e., signal gain) varies with the input frequency of the stimulus  signal.} 
The amplitude response of the system \rvb{represents} the signal gains at all frequencies. Hence, the output amplitude spectrum of a signal is the product of the input amplitude spectrum and the amplitude response of the system. A system is a high-pass (low-pass) filter if the system has a higher amplitude response in the high-frequency (low-frequency) range. The phase response of a system represents the phase shifts of different frequency signals passing through the system.


  \begin{figure}[t] 
    \centering
    \includegraphics[width=3.2in]{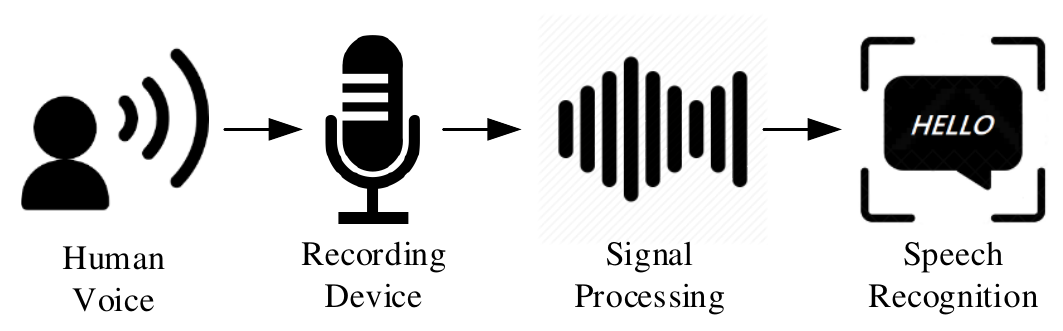}
    \caption{ASR Systems.}
    \label{fig:ASR}
    \vspace{-0.2in} 
    \end{figure}

\vspace{-0.05in}
\subsection{ASR Systems and Replay Attacks}

\noindent \rvb{Figure~\ref{fig:ASR} shows} an automatic speech recognition (ASR) system. A recording device such as a microphone captures  the audio signals from the air \rvb{and converts} the acoustic vibrations into electrical signals. Then, the analog electrical signals are converted to digital signals for \rvb{signal processing}. The processed digital signals are used for speech recognition or speaker identification in the subsequent processing of the ASR systems. These digital signals are commonly referred to as the genuine audio if the signals are directly collected from the live human speakers.

\begin{figure*}[t] 
    \centering
    \subfloat[Classical Replay Attacks]{\includegraphics[width=6.6in]{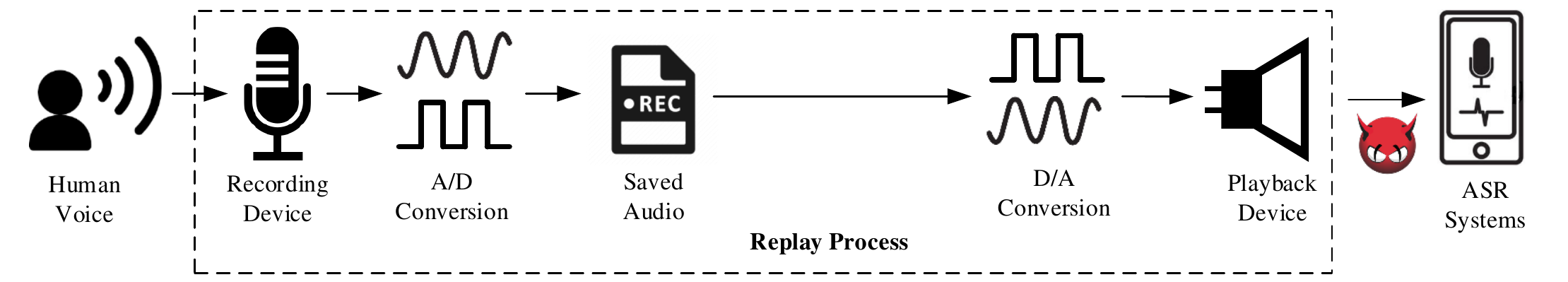}\label{fig:replay}}
    \vfil
    \vspace{-0.1in}
    \subfloat[Modulated Replay Attacks]{\includegraphics[width=6.6in]{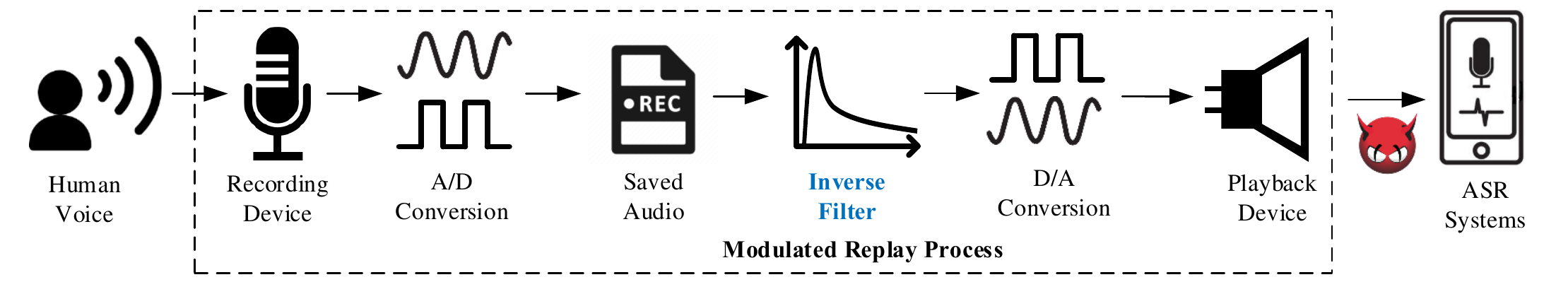}\label{fig:modreplay}}
    \caption{Classical Replay Attacks vs. Modulated Replay Attacks.}
    \label{fig:replaytotal}
    \vspace{-0.1in} 
\end{figure*}


ASR systems are vulnerable to replay attacks. The classical replay attack model contains four basic components, i.e., a recording device, an analog-to-digital (A/D) converter, a digital-to-analog (D/A) converter, and a playback device such as a loudspeaker. Compared with \rvb{the normal speech recognition steps} in the ASR systems, the replay attack \rvb{contains} a replay process as shown in Figure \ref{fig:replaytotal}(a). The attacker firstly collects the genuine human voice using a recording device and  converts the voice to a digital signal by an A/D converter. The digital signal can be stored in a disk device as a lossless compression format or be spread through the Internet. \rvb{After that,} the attacker \rvb{playbacks the digital signal} near the targeted ASR system, \rvb{which spoofs} the system to provide expected services. In the playback process, the stored digital signal is converted to an analog electric signal by a D/A converter. Then, the electric signal will be played as an acoustic wave by a playback device.

\vspace{-0.02in}
\section{Modulated Replay Attacks}\label{sec:attack}
In this section, we propose a new attack called \emph{modulated replay attack}. By analyzing the replay model and replay detection methods, we find the existing defenses rely on the features of amplitude spectrum. Based on these observations, we propose a method to estimate the speaker response and build an inverse filter to compensate the amplitude spectrum of the replay signals. The reconstructed replay signals can bypass the existing defenses.

\vspace{-0.05in}
\subsection{Impacts of Replay Components}

Although classical replay attacks can achieve a high success rate in spoofing ASR systems, some acoustic features can still be utilized to distinguish the replay audio from the genuine audio. As shown in Figure~\ref{fig:replaytotal}(a), the main difference between these two types of audio is the additional replay process that the replay audio goes through.

We study the impacts from four components involved in the replay process, namely, the recording device, A/D converter, D/A converter, and the playback device. We observe that the impacts from the first three components are negligible, and the most significant impacts on replay signals come from the playback device. 
First, an attacker needs to use a recording device to collect the voice command. The main factors that influence the recording process include the non-linearity of modern microphones and the ambient noise. However, the nonlinear frequency range of a microphone is much higher than the human speech frequency. When it comes to the ambient noise, it is hard to tell if the noise is introduced during the attacker's recording process or the ASR recording phase. 

Second, when the A/D converter transforms the signal into a digital form, it may cause the information loss of the analog signal due to the sampling and quantization operations. However, this effect is limited since the modern recording devices have a higher sampling rate (not less than 44.1 kHz) and a higher bit depth (usually higher than 16-bit resolution) than the old-fashioned recorders.

Third, the signal can be transformed back into the analog form by the D/A converter, where a low-pass filter is used to eliminate the high-frequency components caused by sampling. As the sampling frequency is at least 10 times larger than the speech frequency, the filter in the D/A converter has little effect on the audio signals. 

Finally, we find the most significant effects on the replay signal are caused by the playback device. Because of the shape and volume, the acoustic characteristics of loudspeakers are greatly different from those of human vocal organs. Due to the resonance of the speaker enclosure, the voice from loudspeakers contains low-power "additive" noise. These resonant frequency components are typically within 20-60 Hz that human cannot produce~\cite{subbass}.
%
Another important feature of loudspeakers is the low-frequency response distortion due to the limited size of loudspeakers. Within the speech frequency range, the amplitude response of a loudspeaker is a high-pass filter with a cut-off frequency typically near 500 Hz~\cite{phonefreqresp}. As a result, the power of low-frequency components will be attenuated rapidly when a voice signal passes through a loudspeaker, which is the "multiplicative" characteristic of speakers in human speech frequency range~\cite{Frequencyresponse}. Even though the genuine audio and the replay audio have the same fundamental frequency and harmonic frequencies, the power distributions of frequency components remain different. The low-frequency components of replay audio have a smaller power proportion compared with those of genuine audio. Because the different power distributions lead to different timbre, the voice signals sound different even with the same loudness and fundamental frequency.

\vspace{-0.15in}
\subsection{Attack Overview}

Based on our observation that existing defenses utilize the amplitude spectrum to detect replay attacks, the key idea of our proposed attack is to modulate the voice signal so that the replay audio has the same amplitude spectrum as the genuine audio. As shown in Figure \ref{fig:replaytotal}(b), the most critical component is the modulation processor between the A/D and D/A conversion. The modulation processor can compensate for the amplitude spectrum distortion caused by the replay process. By adding the modulation processor, we can deal with the modulated replay process as an all-pass filter, so that the modulated replay audio will have an equivalent processing flow as the genuine audio.

In the classical replay process, the recording device and the A/D and D/A conversion have limited effects on the replay audio. Thus, our modulation processor mainly targets the playback device, specifically, the amplitude response of it. There are many types of playback devices, such as mobile phones, MP3 players, and remote IoT devices in the victim's home. We  acquire the amplitude response of a playback device by measuring the output spectrum in response to different frequency inputs. If the playback device is under remote control that the amplitude response cannot be measured directly, we can estimate an approximate response from the same or similar devices. After acquiring the amplitude response of the playback device, we  design an inverse filter that is a key component in the modulation processor to compensate for the distortion of the signal spectrum. After the spectrum modulation, the modulated replay audio can bypass existing frequency-based defense. 

In our modulated replay attack, the modulation processor only deals with the voice signals in digital form. Therefore, the inverse filter is designed by digital signal processing (DSP) techniques. The modulated signals can be stored or spread through the Internet to launch a remote replay attack.

\begin{figure}[h]
    \vspace{0.05in}
    \centering
    \includegraphics[width=\linewidth]{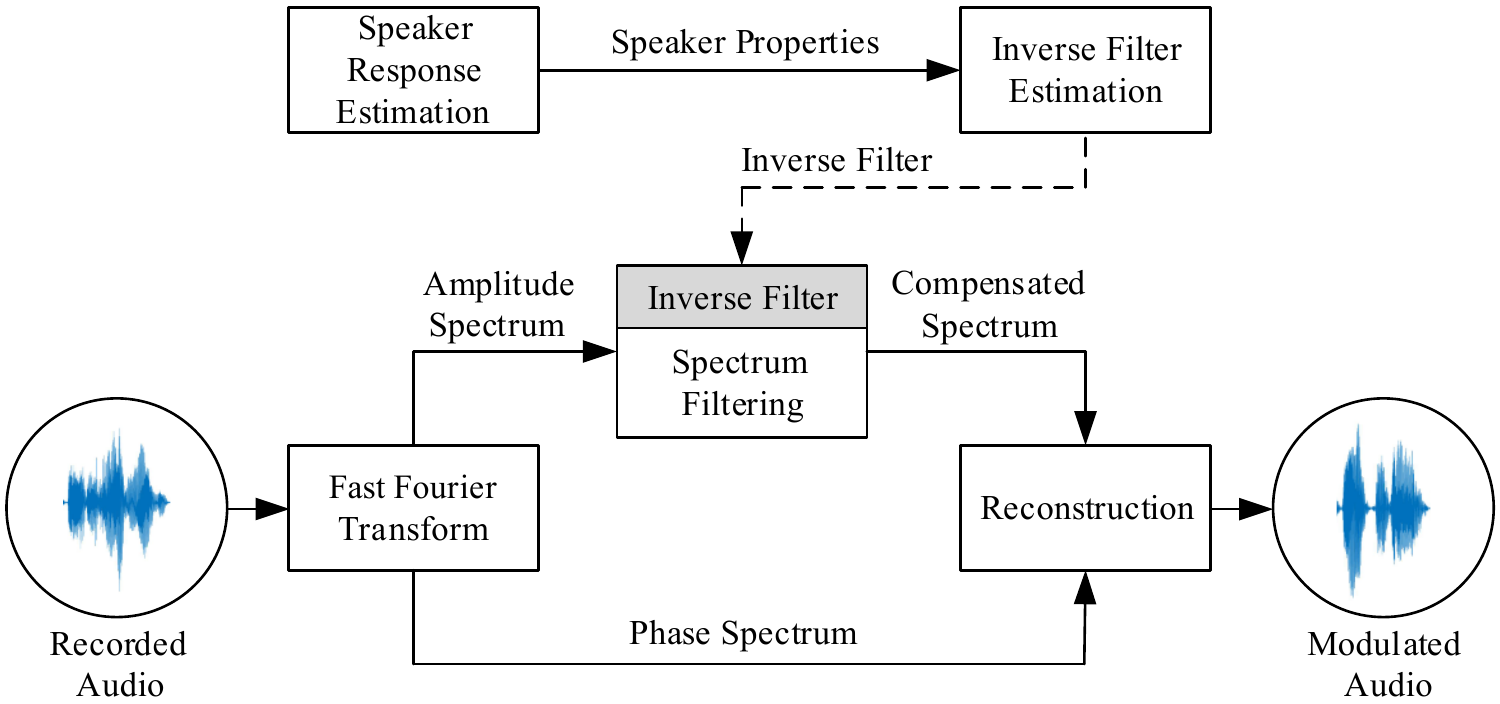}\\
    \caption{The modulation processor.}
    \label{fig_modulated}
    \vspace{-0.15in}
\end{figure}

\subsection{Modulation Processor}

The structure of the modulation processor is shown in Figure~\ref{fig_modulated}. The recorded audio is a digital signal collected from the genuine human voice. The audio is then transformed from the time domain to the frequency domain by fast Fourier transform algorithm. The FFT output is a complex frequency spectrum that can be divided into two parts: (1) the amplitude spectrum that records the amplitude for each frequency component, and (2) the phase spectrum that records the phase angle for each frequency component. We only process the amplitude spectrum in the modulation processor for two reasons. One reason is that both the ASR systems and the replay detection systems extract signal features from the amplitude spectrum. Another reason is that the human ear is less sensitive to the sound phase compared to the sound amplitude. Therefore, the phase spectrum will remain the same in the modulation processor.

The inverse filter, estimated based on the speaker properties, is the key component in the modulation processor. Specifically, the inverse filter is an engine in the spectrum filtering unit, transforming the amplitude spectrum to a compensated spectrum. By the spectrum filtering, the inverse filter can offset the distortion effect caused by the playback device. Therefore, the amplitude responses of the inverse filter and the loudspeaker are complementary, because the combination of these two transfer functions is a constant function that represents an all-pass filter.

After processing the amplitude spectrum with the inverse filter, we can obtain a compensated spectrum that has a better frequency characteristic in the low-frequency range. With both the compensated spectrum and the phase spectrum, the inverse fast Fourier transform (iFFT) is utilized to convert the reconstructed signal from frequency domain to time domain. Finally, we can get a modulated audio in the time domain. Moreover, the modulated audio will be stored as a digital format, which is ready to be used to launch the modulated replay attack.

\subsection{Inverse Filter Estimation}\label{sec:if}
The inverse filter is estimated through the speaker properties. Therefore, it is necessary to measure the amplitude response of the loudspeaker directly. If it is not possible for direct measurement, the amplitude response can be estimated by measuring the speakers in the same or similar model. 

When measuring the speaker properties, we set a single-frequency test signal as the speaker input and record the output audio, as shown in Figure~\ref{fig_speaker}(a). Through checking the output amplitude spectrum, we can get the output amplitude of the corresponding frequency. The amplitude response of the single frequency is the output amplitude divided by the input amplitude. Through changing the input frequency of the test signal, we can obtain the amplitude response over the entire speech frequency range. 

Because the test frequencies of the input signals are discrete, the amplitude response is a series of discrete data points, as shown in Figure~\ref{fig_speaker}(b). To obtain a continuous response function over the entire frequency range, we fill in the missing data by the curve fitting. Cubic spline interpolation~\cite{Spline} will be used to construct a continuous and smooth response curve $H(f)$ with multiple polynomials of degree 3.

As the inverse filter is implemented on the digital signals, we need to convert the continuous response function into a digital form. After the Fourier transform, the signal spectrum has a fixed frequency interval $\Delta f$ denoting the frequency resolution. Hence, we sample the continuous response function at the same frequency interval and get a finer-grained amplitude response. The digital amplitude response of the electronic speaker is denoted as $H(k)$.

\begin{figure}[t]
    \vspace{-0.05in}
    \centering
    \subfloat[The measurement of single-frequency response.]{\includegraphics[width=3.3in]{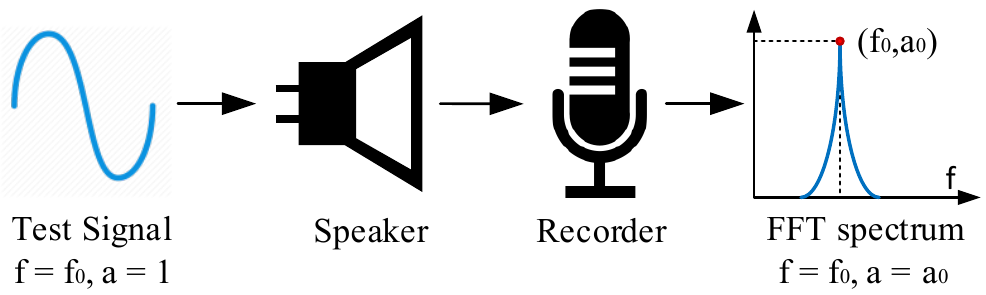}%
    \label{fig_measure}}
    \hfill
    \vspace{-0.1in}
    \subfloat[The processing of fitting speaker response.]{\includegraphics[width=3.3in]{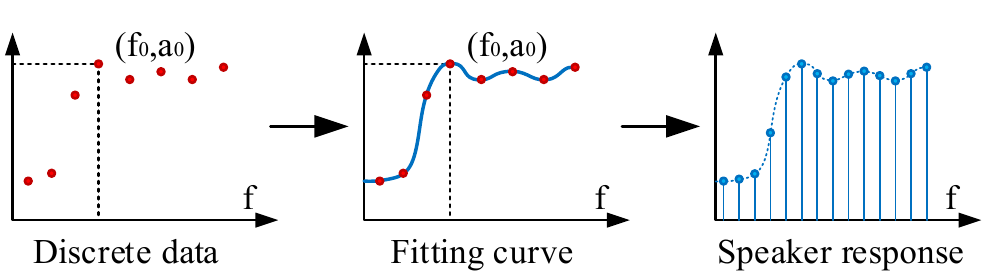}%
    \label{fig_process}}
    \vspace{-0.08in}
    \caption{The method to estimate the speaker response.}
    \label{fig_speaker}
    \vspace{-0.15in}
\end{figure}

After obtaining the speaker amplitude response, we can design the inverse filter by the complementary principle. The amplitude responses of the inverse filter and the speaker can cancel each other, minimizing the impact of the replay process. Hence, the inverse filter $H^{-1}(k)$ should satisfy the all-pass condition that $H^{-1}(k) \cdot H(k) = C$ when $H(k) \neq 0$. $C$ is a positive constant which is typically 1. In addition, if $H(k) = 0$ for any $k$, $H^{-1}(k)$ should also be 0.

Another speaker property is the sub-bass (0-60 Hz) energy, which can be generated by loudspeakers, not humans. The sub-bass features are dependent on the speaker models and enclosure structure~\cite{subbass}. Although attackers may pick the speakers to minimize the sub-bass energy, we still need to minimize the possibility of detected by the sub-bass features. Hence, we optimize the inverse filter in two ways. We set $H^{-1}(k)=0$ when the frequency is within 0-60 Hz, because we do not want to amplify the existing noise in the sub-bass range. Another way is to enhance the inverse filter response in the speech frequency range so as to decrease the relative proportion of the additive sub-bass energy. By these optimizations, we can decrease the metric of sub-bass energy balance under the detection threshold.

By applying the inverse filter before the playback device, we can compensate the unwanted replay effects that are caused by the electronic speakers.

\vspace{-0.05in}
\subsection{Spectrum Processing}

The spectrum processing will involve three phases: the time-frequency domain conversion, the amplitude spectrum filtering, and the modulated signal reconstruction. 

\vspace{-0.03in}
\subsubsection{Time-Frequency Domain Conversion}
First, we need to convert the recorded audio from the time domain into the frequency domain, because it is easier to filter the signals in the frequency domain. For a $L$-length signal segment, we pad the signal with zeros so that the total signal length would be $N$, where $N$ is the smallest power of 2 greater than or equal to $L$. The extended signal is denoted as $x(n), n = 0,1,...,N-1$. Then we convert the time-domain signal $x(n)$ into the frequency-domain representation $X(k)$ through the fast Fourier transform algorithm.

\vspace{-0.1in}
\begin{equation}
X(k) = \sum_{n=0}^{N-1} x(n) \cdot e^{-i 2 \pi k n / N} , k = 0,1,...,N-1   
\end{equation}

$X(k)$ is the frequency spectrum of the original signal in the form of complex numbers. The frequency resolution is defined as the frequency interval $\Delta f = f_{s} / N$, where $f_{s}$ is the sampling rate of the recording audio. 

Then we split the complex frequency spectrum into two parts. The magnitude spectrum $X_{m}(k) = |X(k)|$, represents the signal amplitude of different frequency components $k \cdot \Delta f, k = 0,1,...,N-1$. The phase spectrum $X_{p}(k) = \angle X(k)$ in radians, which is independent with the amplitude information, represents where the frequency components lie in time.

\vspace{-0.03in}
\subsubsection{Spectrum Filtering}

The inverse filter will only be implemented in the amplitude spectrum. The phase spectrum will remain unchanged. The effect of applying a filter is to change the shape of the original amplitude spectrum. According to the system response theory, the compensated amplitude spectrum is the product of the input amplitude spectrum and the amplitude response of the inverse filter. Hence, after modulating the signal with the inverse filter $H^{-1}(k)$, the compensated spectrum $Y_{m}(k)$ satisfies that $Y_{m}(k) = X_{m}(k) \cdot H^{-1}(k)$.

Note that the amplitude spectrum of the speaker output is also the product of the input amplitude spectrum and the speaker amplitude response. Therefore, the amplitude spectrum of the modulated replay audio will be $S_{m}(k) = Y_{m}(k) \cdot H(k)$. We can find that $S_{m}(k) = X_{m}(k) \cdot H^{-1}(k) \cdot H(k) = C \cdot X_{m}(k)$. Because $C$ is a constant, the power distribution of frequency components in the modulated replay audio will be the same as that in the genuine audio, making it harder for ASR systems to detect the replay attack. 

\subsubsection{Modulated Signal Reconstruction} 

After modifying the amplitude spectrum to compensate for the energy loss in the following playback phase, we need to reconstruct the signal in the frequency domain. The modulated signal will have the compensated amplitude spectrum and remain the original phase spectrum. Therefore, the complex frequency spectrum will be reconstructed by the amplitude $Y_{m}(k)$ and the phase angle $X_{p}(k)$. That means the frequency spectrum of the modulated signal should be $Y(k) = Y_{m}(k) \cdot e^{i X_{p}(k)}$ according to the exponential form of complex numbers. After reconstructing the modulated signal in the frequency domain, the complex frequency spectrum $Y(k)$ will be converted back into the time domain by the inverse fast Fourier transform algorithm. 

\vspace{-0.1in}
\begin{equation}
y(n) = \frac{1}{N} \sum_{k=0}^{N-1} Y(k) \cdot e^{i 2 \pi k n / N}, n = 0,1,...,N-1 
\end{equation}

To ensure that the length of the modulated audio is the same as that of the original audio, the last $(N-L)$ data points in $y(n)$ will be discarded. Hence, the total signal length of the modulated audio would be $L$. Then, the final modulated audio will be saved as a digital format to complete the replay attack.

\vspace{-0.05in}
\section{Countermeasure: Dual-domain Detection}\label{sec:defense}

In this section, we propose a countermeasure called \defenseName~  against the modulated replay attack. Due to the similarity of the amplitude spectrum between the modulated replay signals and the genuine signals, the defense will be conducted not only in the frequency domain, but also in the time domain.

\vspace{-0.1in}
\subsection{Defense Overview}

In our scheme, our countermeasure contains two inseparable parts: frequency-domain defense and time-domain defense. A voice command must pass both defenses in time and frequency domains before it can be accepted by ASR systems.

The frequency-domain defense is proved to be effective against classical replay attacks. Because of the frequency spectrum distortion caused by the replay process, we use the power spectrum distribution (timbre) to distinguish the classical replay audio. The area under the CDF curve (AUC) of the power spectrum distribution is extracted as the key frequency-domain feature. We find that the AUC value of the genuine audio is statistically larger than that of the replay audio. By utilizing the frequency-domain defense, we filter out the threat from the classical replay attacks.

The modulated replay audio has the same amplitude spectrum as the genuine audio. Hence, we need to detect the modulated replay audio in other domains. In the phase domain, there is no useful information in the phase spectrum, which records the starting points of each frequency component in the time axis. But in the time domain, we discover and formally prove the following theorem.

\vskip 3pt
\noindent {\bf Theorem. } \emph{There are inevitably spurious oscillations (ringing artifacts) in the modulated replay audio. The amplitude of the ringing artifacts is restricted by the signal amplitude spectrum and absolute phase shifts.}
\vskip 3pt

The mathematical proof of the theorem is demonstrated in Appendix~\ref{append:proof}. In the time domain, based on this theorem, there are small ringing artifacts in the modulated replay signals. However, in the genuine audio and the classical replay audio, the waveform is statistically smooth.

We define a new metric called local extrema ratio to quantitatively describe the strengths of the ringing artifacts. We utilize local extrema ratios at different granularity as the key time-domain feature and filter out modulated replay attacks  using an SVM classifier. 


\vspace{-0.08in}
\subsection{Time-domain Defense}
\label{sec::timedomain}

\begin{figure}[t]
    \vspace{-0.1in}
    \centering
    \subfloat[Coarse granularity ($r=10$)]{\includegraphics[width=1.66in]{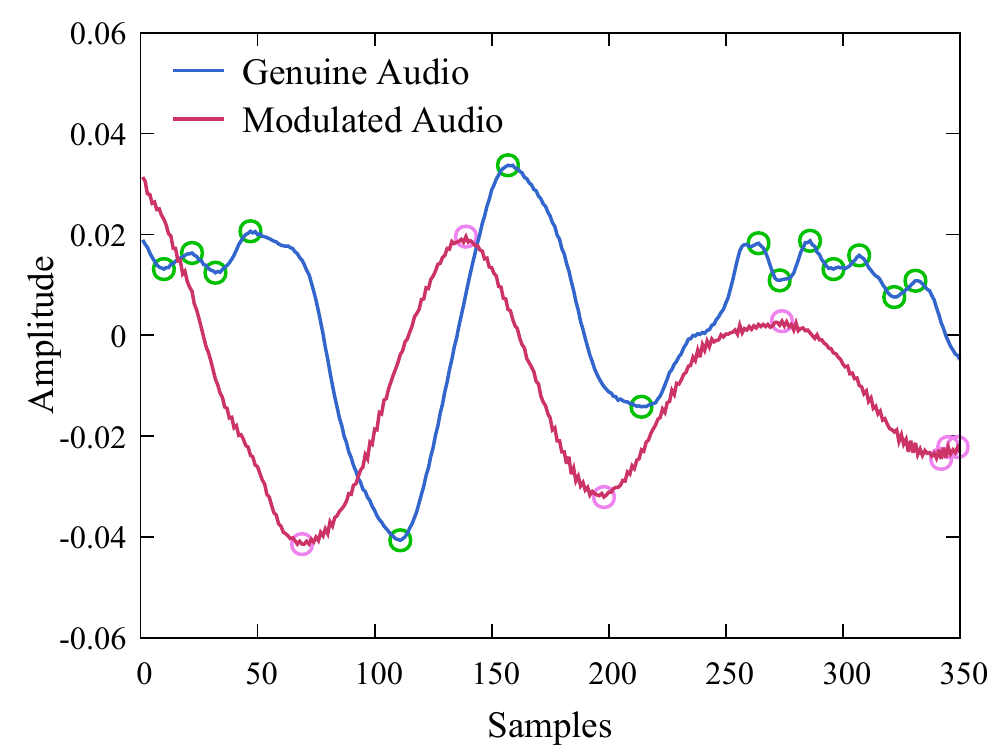}%
    \label{fig_largegran}}
    \subfloat[Fine granularity ($r=1$)]{\includegraphics[width=1.66in]{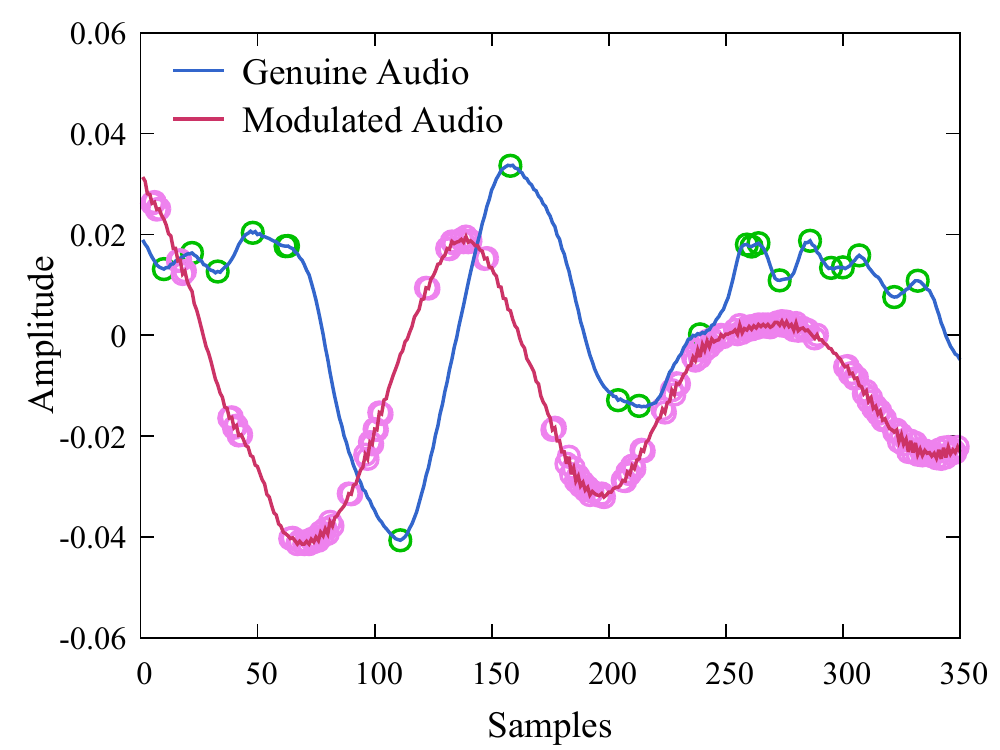}
    \label{fig_smallgran}}
    \vspace{-0.1in}
    \caption{The local extrema under different granularity.}
    \label{fig_granularity}
    \vspace{-0.2in}
\end{figure}  

Because of the difficulty in detecting the modulated replay audio via frequency and phase features, we seek the defenses in the time domain. By our observations and mathematical proof (see Appendix~\ref{append:proof}), we find there are small ringing artifacts in the time-domain signals when performing the modulated replay attack. Although these time-domain artifacts correspond to the high-frequency components, the power of the artifact is too small to be detected in the frequency domain because the maximum amplitude is constraint by the Equation (\ref{eq:constraint}). In the frequency domain, the ringing artifacts can be easily mistaken for the ambient noise. Hence, we propose a time-domain defense method that utilizes the pattern of small ringing artifacts in the modulated replay audio. 

The ringing artifacts pattern is a robust feature that cannot be further compensated by a higher-order filter. The ringing artifacts are caused by the physical property, but not the modulated process itself. When we modulate the recorded audio, there are no ringing artifacts in the processed audio. The ringing artifacts only occur after replaying the processed audio, thus becoming an inevitable feature in the modulated replay audio. In order to describe the ringing artifacts in the time-domain signals, we take \emph{local extreme ratio} as the metric. We firstly give a definition of \emph{local extrema}. 

{\bf Definition:}  In a signal segment $y$, if a sampling point $y_{i}$ is the maximum value or the minimum value in the $(2r+1)$-length window $[y_{i-r}, y_{i+r}]$, $y_{i}$ is a local extrema in the time-domain signal. Note that if the index of the window element is out of bounds, we will pad the window with the nearest effective element.


Local extrema ratio (LER) is defined as the ratio of the local extrema amount to the total signal length. Given an input signal segment, the local extrema ratio correlates with the window parameter $r$. When the window size is small, the LER calculation is in fine granularity that reflects the small ringing artifacts in the time-domain signals. When the window size is large, LER shows the overall change trend of the signals. The modulated replay signals and the genuine signals have different patterns in local extrema ratio with different granularity. We can detect the modulated replay attack via identifying the LER patterns with different parameter $r \in [1, r_{max}]$. Algorithm~\ref{alg:timingdetect} shows the function of obtaining the local extrema patterns and detecting the modulated replay audio.

In Figure~\ref{fig_granularity}(a), under the coarse granularity (larger window size), the number of local extrema does not differ much between modulated replay audio and genuine audio. However, in Figure~\ref{fig_granularity}(b), the situation would be different under the fine granularity (smaller window size). Due to the ringing artifacts, small spurious oscillations occur in modulated replay audio. The number of local extrema in modulated replay audio will be significantly larger than that in genuine audio, which becomes a critical feature that helps us detect the modulated replay attack. A Support Vector Machine (SVM) classifier is trained to distinguish modulated replay audio by determining the local extrema pattern (LEP) with different granularity. The time-domain attack detection is shown in Algorithm~\ref{alg:timingdetect}. The audio will become the candidate audio for the frequency-domain checking if it does not come from the modulated replay attack.

\begin{algorithm}[h]
    \caption{Time-Domain Modulated Replay Detection}
    \label{alg:timingdetect}
    \begin{algorithmic}[1]
      \Require $\text{an~audio~signal}~\bm{y},\text{the largest wnd parameter} ~r_{max}$
      \Ensure $\text{whether there is a modulated replay attack}$
      \State $l \gets length({\bm y})$
      \State $cnt \gets 0$
      \State $\bm{LEP} \gets [~]$
      \For{$r \gets 1$ to $r_{max}$}
        \State $/*~Calculate~Local~Extrema~Ratio~*/$
        \For{$i \gets 1$ to $(l-2)$}
            \State $low \gets max(i-r,0)$
            \State $high \gets min(i+r,l-1)$
            \State ${\bm w} \gets [{\bm  y}_{low},...,{\bm y}_{high}] $
            \If{${\bm w}_{r} = \Call{Min}{\bm w}$ \textbf{or} ${\bm w}_{r} = \Call{Max}{\bm w}$}
            \State $/*~Get~a~Local~Extreme~Point*~/$
            \State $cnt \gets cnt+1$
            \EndIf
        \EndFor
        \State $\bm{LEP}_i = cnt/(l-2)$
      \EndFor
      \State $/*~Identify~Modulated~Replay~Attacks~with~\bm{LEP}~*/$
       
      \If{$SVM\_Classifier(\bm{LEP}) = 1 $}
        \State $\textbf{output}~modulated~replay~attacks$
      \Else
        \State $\textbf{output}~candidate~audio$
      \EndIf
    \end{algorithmic}
  \end{algorithm}


\vspace{-0.15in}
\subsection{Frequency-domain Defense}

The frequency-domain defense is used to counter the classic replay attack. It is based on the noticeable different timbre of the voice sounded from human and electronic speakers. 

In the replay model, each component frequency in the genuine audio is exactly the same as that in the replay audio, no matter the fundamental frequency or the harmonics. For example, if the fundamental frequency of the genuine audio is 500 Hz, the replay audio will also have a fundamental frequency of 500 Hz. However, even with the same component frequencies, the genuine human voice and the replay voice sound different in our perception. The main reason is the power distributions of the frequency components, namely the timbre, are different. 

For human, our voice is sounded from the phonatory organ. The typical sound frequency for human is within the range from 85 Hz to 4 kHz, where the low-frequency components are dominant. For electronic speakers, there is an acoustic defect on the low-frequency components due to the speaker structure, materials, and the limited size. The power of the replay signals decays dramatically in the low-frequency range, especially under 500 Hz. Meanwhile, the human fundamental frequency range is 64-523 Hz for men, and 160-1200 Hz for women. Hence, the electronic speakers will attenuate the power in the human fundamental frequency because of the speaker properties. With respect to the power distribution, the power of the genuine audio is mainly concentrated in the low-frequency range, while the power of the replay audio is more distributed in the speech frequency range. Our frequency-domain defense utilizes these timbre features to defeat the classic replay attack. 

\begin{figure}[t]
  \centering
  \includegraphics[width=2.75in]{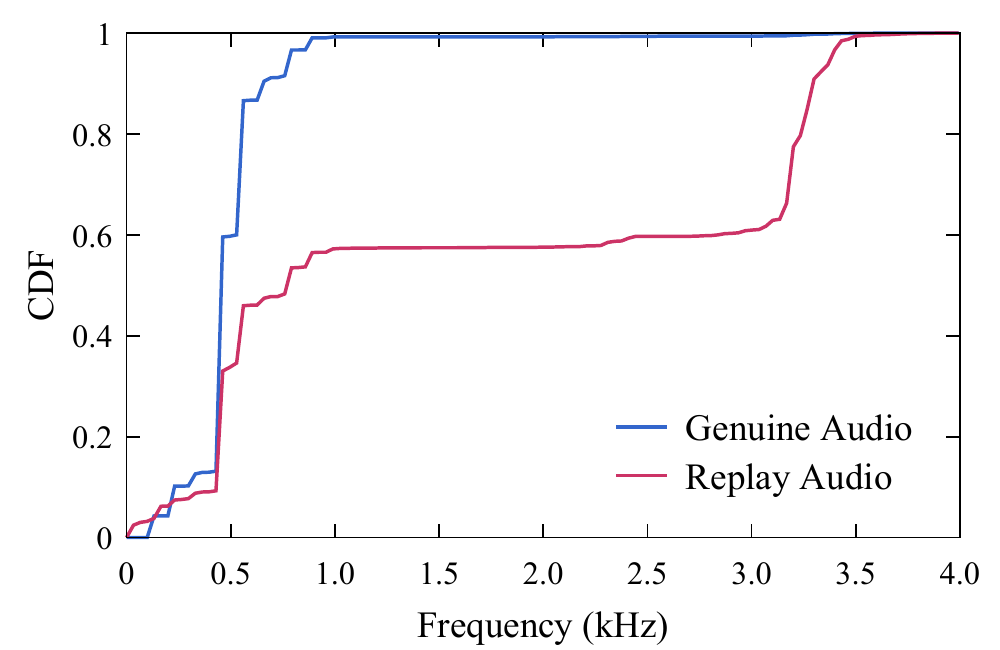}\\
  \caption{Cumulative density function of spectral power distribution for genuine and direct replay audios.}
  \label{fig_cdf}
  \vspace{-0.15in}
\end{figure}

\begin{algorithm}[h]
    \caption{Frequency-Domain Replay Detection}
    \label{alg:fredetect}
    \begin{algorithmic}[1]
      \Require $\text{an audio signal}~\bm{y},~\text{FFT point numbers} ~N,~~~~~~~~~$ $\text{decision threshold}~A_{th}$
      \Ensure $\text{whether there is a classical replay attack}$
      \State $/*~Calculate~Normalized~Signal~Power~Spectrum*/$
      \State $\bm{K} \gets FFT({\bm y},N)$
      \State $p \gets \sum_{i=0}^{N-1}{\bm{K}_{i}^2}$
      \For{$i \gets 0$ to $N-1$}
        \State $\bm{D}_i = {\bm{K}_i^2}/{p}$ 
      \EndFor
      \State $/*~Calculate~the~CDF~and~its~AUC~*/$
      \State $\bm{A}_0 = \bm{D}_0$
      \For{$i \gets 1$ to $N-1$}
        \State $\bm{A}_i = \bm{A}_{i-1} + \bm{D}_i$ 
      \EndFor
      \State $AUC = \sum_{i=0}^{N-1}{\bm{A}_i}/N$
      \State $/*~Identify~Classical~Replay~Attacks~with~AUC*/$
      \If{$AUC~<~A_{th} $}
        \State $\textbf{output}~replay~attacks$
      \Else
        \State $\textbf{output}~genuine~audio$
      \EndIf
    \end{algorithmic}
  \end{algorithm}

Timbre is described by the power distribution of different frequency components. It is necessary to define a mathematical description for the timbre. When an ASR system captures a voice signal from the air with a sampling rate of $f_{s}$, we firstly obtain the amplitude spectrum of the signal through $N$-point fast Fourier transform. The signal amplitude spectrum is denoted as $K(n)$, $n=0,...,N-1$, with the frequency resolution $\Delta f = f_{s} / N$. The frequency value of the $i$-th component is $i \cdot \Delta f$, while the amplitude is $K(i)$. 
Hence, the signal power spectrum is $K^{2}(n)$, and the power spectral density (PSD) of frequency components is defined as $D(n) = K^{2}(n) / \sum_{i=0}^{N-1}K^{2}(i)$. To distinguish the different power distributions, we measure the cumulative density function (CDF) $A(n)$ for the power spectral density, 

\vspace{-0.1in}
\begin{equation}
A(n) = \sum_{i=0}^{n}D(i) =  \sum_{i=0}^{n}K^{2}(i) / \sum_{i=0}^{N-1}K^2(i).
\end{equation}

$A(n)$ is a monotonically increasing function, with a range of $[0,1]$. As shown in Figure~\ref{fig_cdf}, the power spectrum CDF of genuine audios and replay audios are quite different. For  genuine audios, the power is concentrated in the low-frequency range, so the CDF rises more quickly. For replay audios, the CDF function grows slower due to the more distributed power spectrum. We utilize the CDF characteristic to distinguish replay audios from genuine audios.

We utilize the area under the CDF curve (AUC) to verify and filter out the classic replay audio. AUC is calculated as $\sum_{n} A(n) / N$. If the AUC value is less than a specific threshold $A_{TH} \in (0,1)$, there is a classic replay attack. We show the frequency-domain attack detection in Algorithm~\ref{alg:fredetect}.

\vspace{-0.05in}
\subsection{Security Analysis}
We discover and prove that there are inevitably either ringing artifacts in the time domain or spectrum distortion in the frequency domain, no matter if replay signals are modulated.

For the frequency-domain defense, the principle comes from the signal difference of the power spectrum distributions. It is known that human speech is not a single-frequency signal, but a signal with fundamental frequency $f$ and several harmonics $nf, n \geq 2$. Within the human voice frequency range, the speaker response has a great difference in the low-frequency band and the high-frequency band, which means $H(f) \neq H(nf)$. As a result, the power ratio of genuine audio $A(f)/A(nf)$ is different from that of the corresponding replay audio $(H(f) \cdot A(f))/(H(nf) \cdot A(nf))$. The different power ratios cause the difference in the power spectrum distributions.

For the time-domain defense, we can prove that there are inevitably spurious oscillations (ringing artifacts) in the modulated replay audio. The critical factor is the inevitable phase shifts that cannot be accurately measured (see details in Appendix~\ref{append:proof}). Although the amplitude spectrums are the same, the signal phase spectrums can be different. The relationship between the amplitude spectrum to the time-domain signals is a one-to-many relationship. Moreover, we cannot compensate for the phase shifts due to the limits of the accuracy in measurements. Even a small phase error can cause ringing artifacts in the time-domain. That is why we need to check the signals in both frequency domain and time domain.

Besides, the high local extrema ratio in the modulated replay audio can result from other aspects, i.e. the measurement error, the FFT truncation effect, and the time-domain joint. 
First, the measurement involves exponential computation, where the round-off errors can be accumulated so that the amplitude estimation is not accurate, finally bringing about parasitic vibration in the modulated replay signals.
Second, the real FFT operation works on a finite-length signal, which is equivalent to adding a window function to an infinite-length signal. The window function in the time domain corresponds to a $sinc(x)$ function convolved in the frequency domain, causing the frequency spectrum to expand and overlap.
Third, when splicing the reconstructed signals into new audio, there is no guarantee of the continuity at the starting and ending splice points. A discontinuous splice point can lead to ringing artifacts due to the Gibbs phenomenon~\cite{GibbsPh}.

Moreover, ringing artifacts cannot be further compensated by a higher-order filter since ringing artifacts only occur after the replay process rather than after the modulation process. Moreover, the iterative filtering scheme can reduce ringing artifacts in image restoration that are mainly caused by overshoot and oscillations in the step response of an image filter~\cite{Ringing}. However, it is not suitable for speech signals because the ringing artifacts are introduced by hardware properties. Even if attackers might reduce ringing artifacts to a certain extent, the time-domain defense can still detect modulated replay audio. This is because our method does not rely on the amplitude threshold of ringing artifacts. Although the amplitude of ringing artifacts may decrease, the local extrema cannot be eliminated. The time-domain defense uses local extrema as features so that even small ringing artifacts can be detected.

\section{Evaluation}\label{sec:eval}

In this section, we conduct experiments in a real testbed to evaluate the modulated replay attack and our defense.

\begin{figure}[h]
  \vspace{-0.05in}
  \centering
  \includegraphics[width=3.3in]{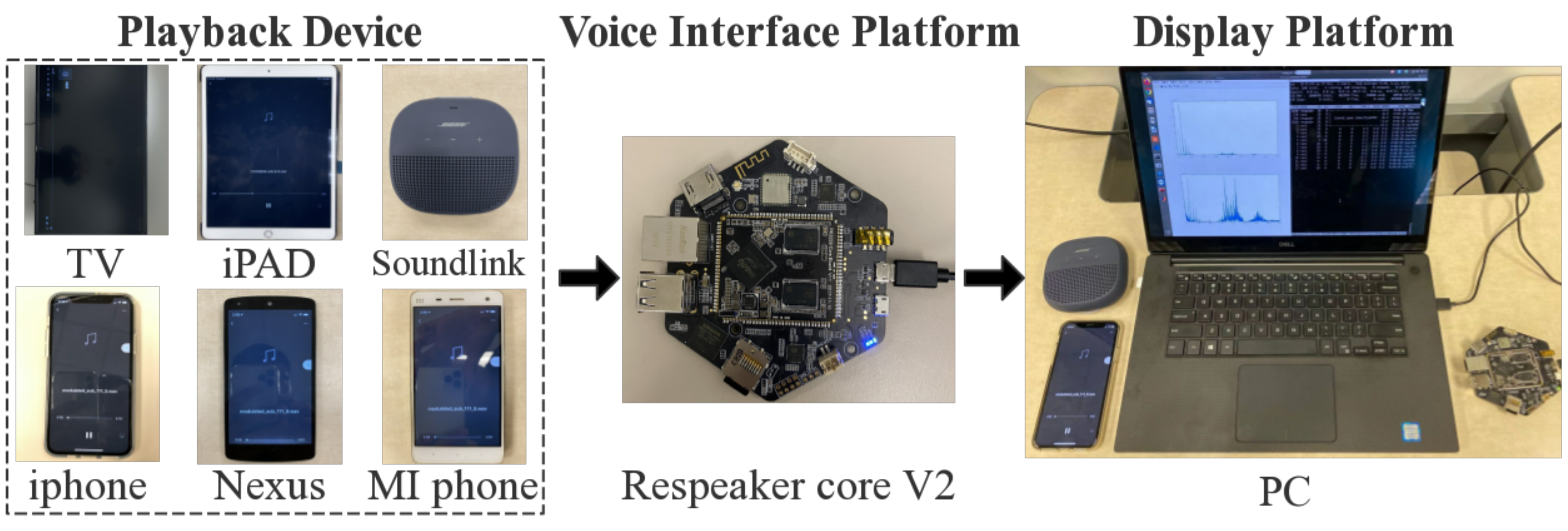}
  \caption{The testbed in our experiments.}
  \label{fig_testbed}
  \vspace{-0.15in}
\end{figure} 

\subsection{Experiment Setup}
\label{sec::evaluation}
We use a TASCAM DR-40 digital recorder for collecting the voice signals. The sampling rate of the digital recorder is set to 96 kHz by default.
{We conduct real experiments with} a variety of  {common electronic devices in our lives}, including iPhone X, iPad Pro, Mi Phone 4, Google Nexus 5, Bose Soundlink Micro, and Samsung UN65H6203 Smart TV.  {Figure~\ref{fig_testbed} shows the testbed in our experiments.}
We aim to demonstrate that both 
our attack and countermeasure scheme can be applied to  {various} speaker devices.
{To generate modulated replay audios, we apply MATLAB to estimate the amplitude response and design the inverse filter for different speakers. Due to space constraints, we put the details in Appendix~\ref{append:spk}.

ASVspoof 2017~\cite{ASV2017} and ASVspoof 2019~\cite{ASVspoof2019} are two popular databases for replay attacks. However, we cannot convert the replay attack samples in these two databases into modulated replay attacks, due to the lack of information of replay devices. Instead, to conduct a fair comparison between modulated replay audio and classic replay audio, we collect an in-house dataset with 6 replay devices. For each of these replay devices, the dataset contains 222 modulated replay audios as well as 222 corresponding classic replay audios. All audio signals are collected in a quiet lab environment. We use 10-fold cross-validation accuracy as a metric since it can reflect the whole performance of the system.
Moreover, we implement the prototype of our defense \defenseName~ in C++ language and run it on a popular voice interactive platform, i.e., ReSpeakerCore v2.}

\subsection{Effectiveness of Modulated Replay Attacks}
\label{sec::exp_mod}

\begin{figure}[t]
    \centering
    \subfloat[Genuine Audio Collected from Human]{\includegraphics[width=2.7in]{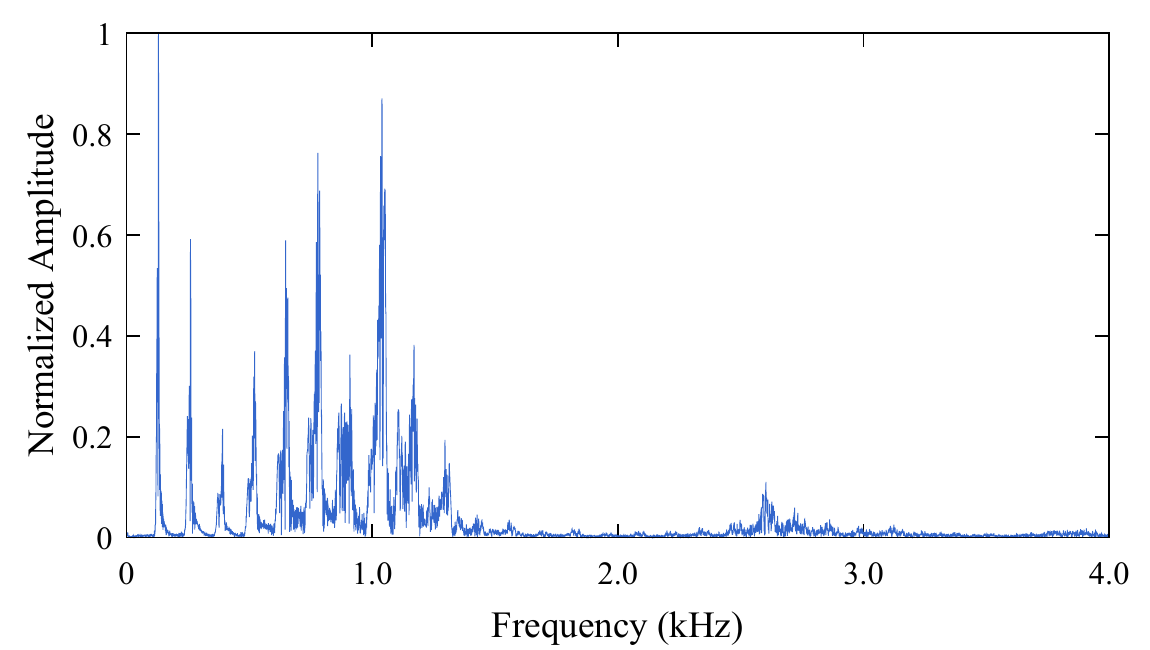}}
    \hfil  
    \subfloat[Direct Replay Audio]{\includegraphics[width=2.7in]{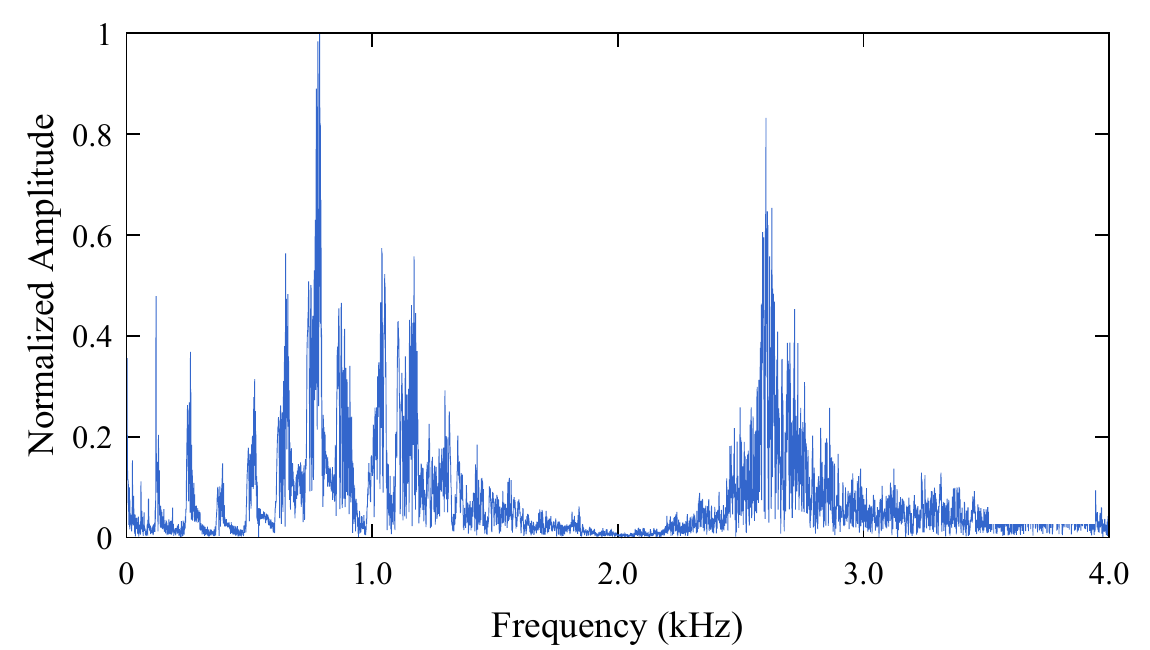}} 
    \hfil
    \subfloat[Modulated Replay Audio]{\includegraphics[width=2.7in]{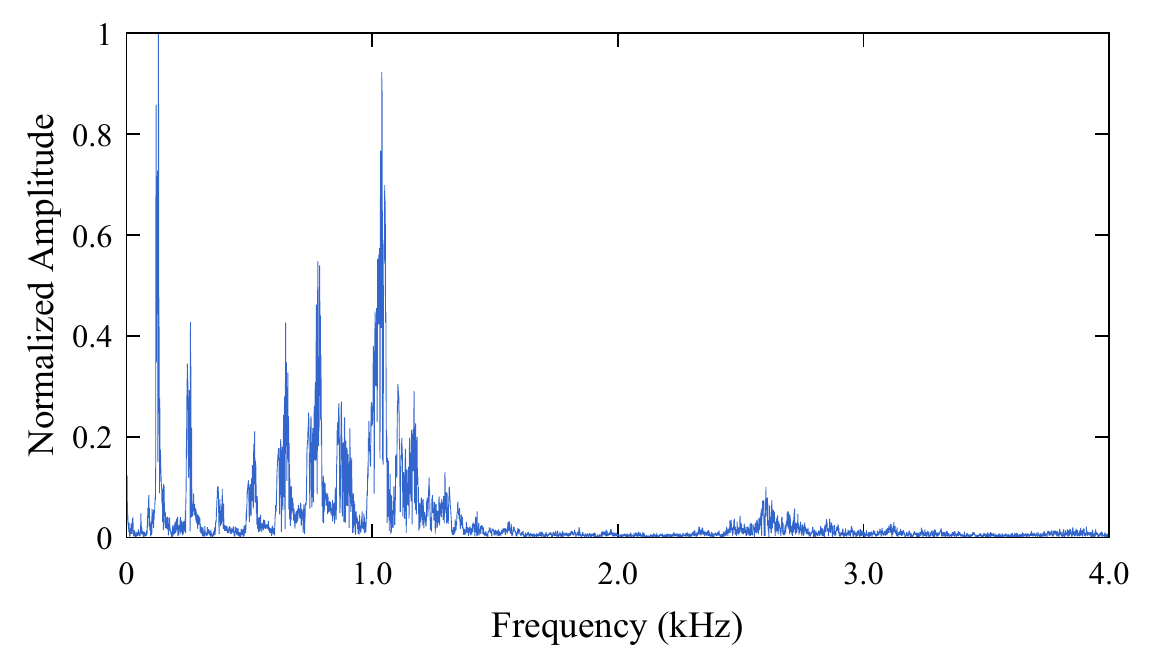}}
    \caption{Amplitude spectrum of audio signals.}
    \label{fig_exp1}
    \vspace{-0.18in}
\end{figure}

We conduct experiments with the modulated replay attack.  {The attack leverages the inverse filter} to generate synthetic audio that has a similar frequency spectrum as the genuine audio. The modulated signals are generated in the Matlab environment and stored in a lossless format. They are then transferred to replay devices for performing attacks. Figure~\ref{fig_exp1} shows the amplitude spectrum of the signals during the modulated replay process in our experiments. 
{Here, the results are collected} using the iPhone device, while we have similar results with other devices. Figure~\ref{fig_exp1}(a) illustrates the genuine audio that is captured directly from a live human in a quiet room environment. The energy of genuine audio is mainly concentrated in the low-frequency range. Figure~\ref{fig_exp1}(b) shows the spectrum of the direct replay audio, which is captured from the direct playback of the genuine audio. Due to the response properties of the speaker devices, the high-frequency components in the direct replay audio have a higher relative proportion compared with those in the genuine audio. The spectrum difference is a vital feature in the various classic replay detection methods.

{Figure~\ref{fig_exp1}(c) shows the spectrum of the modulated replay audio collected by the ASR system. We can see that the low-frequency energy is greatly enhanced to cope with the speaker effects. Thus, the spectrum of the modulated replay audio is very similar to that of the genuine audio in Figure~\ref{fig_exp1}(a).} 
{Moreover, we quantify the similarity between the modulated replay audio and the genuine audio using the L2 norm comparison~\cite{L2Norm} that has been widely used to compare the spectrums of audio.} It is defined as $\left \| K_{1} - K_{2} \right \|_{2}^{2}$, where $K_{1}$ and $K_{2}$ are  {two normalized spectrum distributions of audio}, 
and $\left \| \cdot \right \|_{2}^{2}$ is the square of Euclidean distance.  {The smaller the L2 norm is, the more similar the two audios are.}  We measure the similarity values on 660 pairs of audio samples, the average similarity between the modulated replay audio and the genuine audio  is $1.768 \times 10^{-4}$. However, the average similarity between the direct replay audio and the genuine audio $S_{rg}$ is $15.71 \times 10^{-4}$ on average,  {which is much larger than the similarity between the modulated replay audio and the genuine audio}.  {The results demonstrate that the modulated replay audio is much more similar to the genuine audio.}

{Furthermore, we re-implement 8 popular detection methods that can be divided in three categories, namely, Cepstral Coefficients Features based defense, High-frequency Features based defense, and Low-frequency Features based defense. We apply those defense methods to detect both direct replay attacks and modulated replay attacks on 6 electronic devices, and the results in Table~\ref{tab:comp} show that our modulated replay attacks can bypass all these countermeasures.}


\vspace{0.04in}
\noindent \textbf{Bypassing Cepstral Coefficients Features Based Defense. }
The most popular method  {to detect replay attacks} is based on cepstral coefficients features extracted from the signal amplitude spectrum. These cepstral coefficients features includes CQCC~\cite{CQCC}, MFCC~\cite{MFCC}, LPCC~\cite{Lav2017}, and MWPC~\cite{LPCC}.  {Our experiments show that the accuracy of detecting direct replay attacks} can always achieve over 88\% accuracy.  {However, Table~\ref{tab:comp} shows the accuracy significantly drops to 1.80\%$\sim$58.56\%} when detecting the modulated replay audio.   {The results indicate that our modulated attack can  bypass existing cepstral coefficients based detection methods.}

\begin{table*}[t]
    \centering
    \caption{\label{tab:comp} The accuracy of different defense methods on detecting direct replay attacks and modulated replay attacks.}
    \vspace{-0.1in}
    \renewcommand{\arraystretch}{1.3}
    \begin{threeparttable}[b]
      
        \begin{tabular}{|c|c|c|c|c|c|c|}
            \hline
            \textbf{Detection Method\tnote{$\dagger$}}  & {iPhone}    & {iPad}             & {Mi Phone}           & {Google Nexus}        & {BOSE}            & {Samsung TV} \\
            \hline
            \hline
            {CQCC}~\cite{CQCC}                  &  95.95\% / 4.50\%\tnote{$\star$}  &  95.51\% / 6.31\%  &  92.18\% / 8.11\%  &  {89.93\% / 2.25\%}  &  91.90\% / 7.21\%  &  95.51\% / 6.76\%   \\ 
            \hline  
            {MFCC}~\cite{MFCC}                  &  {90.99\% / 15.51\%}  &  93.24\% / 18.92\%  &  89.64\% / 24.32\%  &  89.19\% / 27.03\%  &  91.89\% / 29.73\%  &  90.99\% / 27.71\%   \\        
            \hline
            {LPCC}~\cite{Lav2017}               &  {89.19\% / 8.11\%}  &  87.84\% / 9.91\%  &  90.09\% / 15.32\%  &  86.03\% / 18.92\%  &  87.84\% / 11.71\%  &  90.54\% / 11.26\%   \\ 
            \hline
            {MWPC}~\cite{LPCC}                  &  95.05\% / 46.85\%  &  {92.79\% / 36.04\%}  &  90.99\% / 53.15\%  &  95.05\% / 43.24\%  &  100.0\% / 50.45\%  &  86.93\% / 58.56\%   \\ 
            \hline
            \hline
            {Sub-band Energy}~\cite{energy_sep} &  89.61\% / 5.41\%  &  89.22\% / 4.50\%  &  89.70\% / 6.31\%  &  88.61\% / 10.81\%  &  {84.11\% / 0.00\%}  &  85.57\% / 0.90\%   \\ 
            \hline
            {HF-CQCC}~\cite{freq_solution}      &  90.91\% / 25.23\%  &  90.91\% / 22.52\%  &  90.91\% / 24.32\%  &  90.08\% / 18.02\%  &  93.94\% / 38.74\%  &  {93.94\% / 11.71\%}  \\ 
            \hline
            {FM-AM}~\cite{ModDynamic}           &  92.86\% / 7.21\%  &  92.86\% / 17.12\%  &  {89.29\% / 4.5\%}  &  92.86\% / 9.91\%  &  92.86\% / 35.14\%  &  96.43\% / 12.61\%   \\ 
            \hline
            \hline
            {Sub-bass}~\cite{subbass}           &  99.10\% / 7.66\%  &  {99.10\% / 4.50\%}  &  98.20\% / 5.80\%  &  98.65\% / 4.95\%  &  96.85\% / 6.76\%  &  97.30\% / 5.40\%   \\ 
            \hline
            \hline
            {DualGuard}                         &  91.00\% / 98.88\%  &  90.54\% / 98.32\%  &  {89.19\% / 97.75\%}  &  90.45\% / 98.22\%  &  90.10\% / 97.79\%  &  89.64\% / 99.65\%  \\
            \hline
        \end{tabular}
        
        \begin{tablenotes}
            \item[$\dagger$] The parameters of the different detection methods are listed in Appendix \ref{append:par}.
            \item[$\star$] The first number is on  direct replay attack and the second number is on modulated replay attack.  
        \end{tablenotes}
        
    \end{threeparttable}
    \vspace{-0.15in}
\end{table*}

\vspace{0.04in}
\noindent\textbf{Bypassing High-frequency Features Based Defense. }
As shown in Figure~\ref{fig_exp1}(a) and Figure~\ref{fig_exp1}(b), the high-frequency spectral features between the genuine audio and the replay audio are significantly different. Therefore, a number of methods~\cite{energy_sep,ModDynamic,freq_solution} detect replay attacks using high-frequency features, including Sub-band Energy~\cite{energy_sep}, HF-CQCC~\cite{freq_solution}, and FM-AM~\cite{ModDynamic}. Table~\ref{tab:comp} shows they can achieve high accuracy on detecting the direct replay attack, e.g., 96.43\%. However, they fail to detect the modulated attack due to frequency compensation. The highest accuracy on detecting the modulated replay attack is only 38.74\%.

\vspace{0.04in}
\noindent \textbf{Bypassing Low-frequency Features Based Defense. }
Besides detection based on high-frequency features, a recent study~\cite{subbass} provides an effective method, i.e. Sub-bass, to detect replay attacks based on low-frequency features. It defines a metric named \emph{energy balance metric}, which indicates the energy ratio of the sub-bass range (20-80 Hz) to the low-frequency range (20-250 Hz). Our experiments show that it can achieve 99.1\% accuracy on detecting direct replay attacks with the metric. However, the accuracy significantly drops to less than 8\% when detecting modulated replay attacks.

In these 8 detection methods above, MWPC performs better than other techniques. This is because MWPC can capture partial temporal information using the mel-scale Wavelet Package Transform (WPT)~\cite{WPT}, which handles the temporal signals on different scales. HF-CQCC can capture the high-frequency difference in signals. Such partial temporal information and high-frequency difference provide more useful features for the detection of replay attacks. Thus, MWPC and HF-CQCC perform better than other techniques. In addition, Table~\ref{tab:comp} also shows the experimental results of the modulated replay attack with six loudspeaker devices respectively. In theory, whatever frequency response a speaker has, we can always find the corresponding inverse filter to counteract the effect of the replay process. As a result, the modulated replay attack does not depend on any specific type of speaker. The experimental results in Table~\ref{tab:comp} validate our attack design. For any specific detection method, the modulated replay attack exhibits similar performance when leveraging different speaker devices. This property is critical for real-world replay attacks, because it demonstrates the modulated replay attack is independent of the loudspeaker. An attacker can utilize any common speaker in our lives to perform the modulated replay attack against ASR systems.


\vspace{-0.05in}
\subsection{Effectiveness of Dual-Domain Detection}
\label{sec::defense}

Our defense, i.e. \defenseName{}, contains two parts: time-domain detection and frequency-domain detection. The time-domain detection mainly aims to identify modulated replay attacks and the frequency-domain detection mainly aims to identify direct replay attacks. We show the experimental results for these two parts, respectively. 

\vspace{0.0in}
\noindent \textbf{Time-Domain Detection.} 
We  {conduct experiments to evaluate the accuracy for \defenseName~to detect modulated replay attacks in the time domain. As the local extrema ratio (LER) is the key feature to detect replay attacks in the time domain, we first measure the LER values of both modulated replay audios and genuine audios from 6 different speaker devices.  }

Figure~\ref{fig_timeexp} illustrates the change of LER value from fine granularity (with small window size) to coarse granularity (with large window size). 
{We can see that the LER decreases with the increase of  the window size. When the window size is small}, the LER value of the modulated replay audio is statistically larger than that of the genuine audio, which is the main difference between these two types of audios. As we mentioned in Section~\ref{sec::timedomain}, the relatively high LER value results from the ringing artifacts in the modulated replay audio.  {The results demonstrate the feasibility to detect the modulated replay attack in the time domain with the LER patterns.}

\begin{figure}[h]
 \vspace{-0.1in}
  \centering
  \includegraphics[width=2.75in]{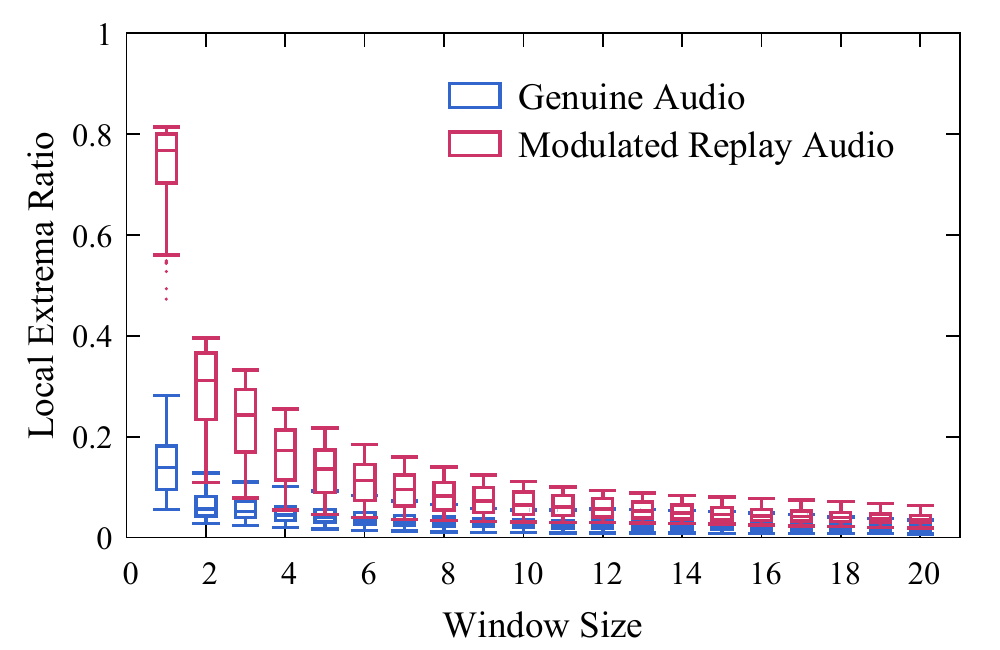}
  \vspace{-0.05in}
  \caption{20-dimensional local extrema patterns with different granularity for genuine and modulated replay audios.}
  \label{fig_timeexp}
  \vspace{-0.05in}
\end{figure}

{We conduct experiments to evaluate the detection accuracy in the time domain with Algorithm~\ref{alg:timingdetect}. As shown in Figure~\ref{fig_timeexp}, there are no significant differences for the LERs of the genuine audio and the modulated replay audio when the window size reaches 20. Thus, we choose a 20-dimensional tuple $\{$LER$_{1}$, LER$_{2}$, ..., LER$_{20}\}$ in our algorithm} as the feature to detect the modulated replay attack. Here, LER$_{r}$ denotes the LER value with the window size $r$.  {The detection accuracy of \defenseName~ on modulated replay attacks is shown in Table~\ref{tab:comp}. We can see that \defenseName~can accurately identify modulated replay attacks in the time domain. The detection accuracy for modulated replay attacks always exceeds  97\% with different speakers. We also calculate the false positive rate of our method in detecting modulated replay attacks. It always maintains less than 8\% false positive rate. The results demonstrate the generalization ability of \defenseName~with different speakers. Actually, the  generalization is due to the robust artifact properties in the time-domain signals (see Appendix~\ref{append:proof}). Our time-domain defense is independent of speakers. 
}
Our main contribution of time-domain defense is on the key feature extraction. For the experiments on comparing different classifiers, we refer the readers to Appendix~\ref{append:classifier}. In our defense, we choose SVM due to its high performance and easy deployment.

\vspace{0.03in}
\noindent \textbf{Frequency-Domain Detection.} We  {conduct experiments to evaluate the accuracy for \defenseName~to detect direct replay attacks in the frequency domain. To decide the decision threshold of Algorithm~\ref{alg:fredetect}, we first obtain the Area Under CDF curve (AUC) from the amplitude spectrum of audios. }
Figure~\ref{fig_auc} shows the AUC distributions for both genuine audios and direct replay audios. We can see that the AUC values of  genuine audios are concentrated and close to 1, which indicates that the low-frequency energy  {is dominant}. However, the AUC values of direct replay audios are distributed and small, which is consistent with the distributed spectrum of replay audios. 
{As shown in Figure~\ref{fig_auc}, the best decision threshold is 0.817 since it can minimize the classification errors between genuine audios and replay audios.}
{Table~\ref{tab:comp} shows the detection accuracy of \defenseName~ on direct replay attacks using Algorithm~\ref{alg:fredetect} with a decision threshold of 0.817. The accuracy with different speakers always exceeds 89\%. We also calculate the false positive rate of our method in detecting direct replay attacks. It always maintains less than 5\% false positive rate. } Moreover, we conduct experiments with the ASVspoof 2017 and 2019 datasets to show that DualGuard can effectively detect classic replay attacks. Our experimental results show that DualGuard can achieve 87.13\% and 83.80\% accuracy in these two datasets, respectively. 


\begin{figure}[t]
  \centering
  \includegraphics[width=2.75in]{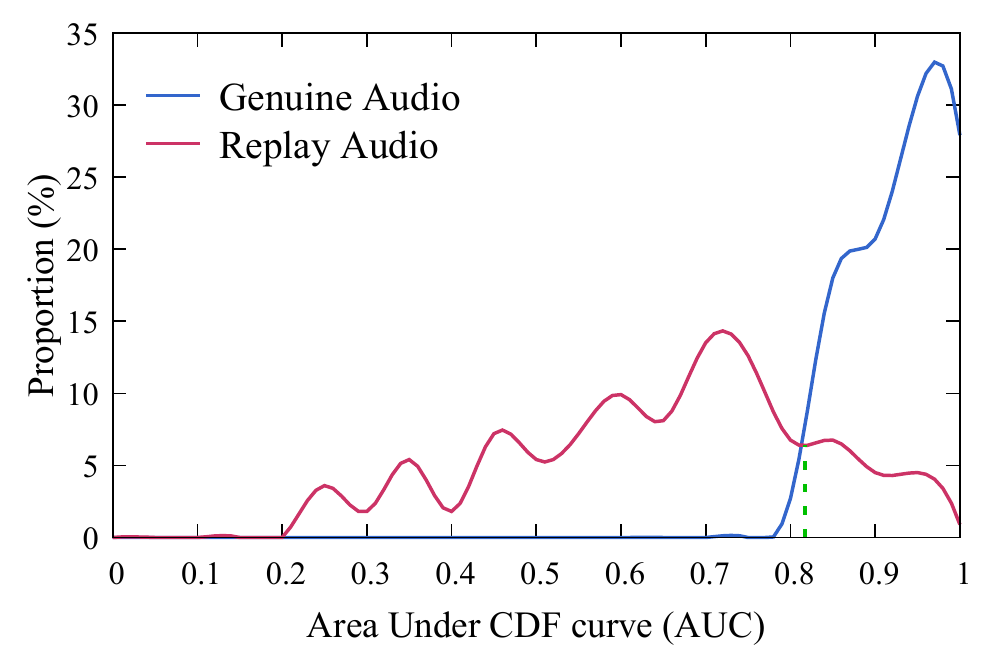}
  \caption{The AUC distribution of genuine audios and replay audios with the classification decision threshold.} 
  \label{fig_auc}
  \vspace{-0.2in}
\end{figure} 

{Moreover, we train another model only with frequency features from a mix of genuine audios, direct replay audios, and modulated replay audios in order to demonstrate the necessity to detect all replay attacks in two domains. Our experimental results show that the accuracy can only reach 63.36\%. It is due to the great spectral similarity of genuine audios and modulated replay audios in the frequency domain. Therefore, the dual-domain detection is necessary to accurately detect both two types of replay attacks.}



\vspace{-0.05in}
\subsection{Robustness of Dual-Domain Detection}\label{sec::robustness}
We conduct experiments to show the robustness of our dual-domain detection under different sampling rates, different recording devices, different speaker devices, and different noisy environments.

\vspace{0.03in}
\noindent \textbf{Impact from Genuine Audio Sampling Rate.} We evaluate the impact of the sampling rate for recording the initial human voice by attackers. We first use TASCAM DR-40 digital recorder with fs = 96 kHz to capture 
initial human voice. We also use iPhone X with fs = 48 kHz to capture human voice. For both sampling rates, the average detection accuracy of \defenseName{} on modulated replay attack is 98.05\%. That is because the sampling rate used by attackers only changes the spectral resolution in the modulation process. However, the waveform of modulated replay audios will not be changed since D/A converter
will convert modulated signals into analog form before the replay process. 

\vspace{0.03in}
\noindent \textbf{Impact from ASR Sampling Rate.} 
We conduct experiments on different recording devices with different sampling rates. In our experiments, there are three settings of  {sampling rates} for our  recording devices: (S1) TASCAM DR-40 with 96 kHz, (S2) TASCAM DR-40 with 48 kHz, and (S3) a mobile phone (Xiaomi 4) with 44.1 kHz. 
 {Figure \ref{fig_detectacc}(a) shows the experimental results.}  
We can see the detection accuracy usually  {increases with the increase of sampling rates. We find that although} changing the sampling rate has little effect on the frequency-domain detection, it significantly affects the time-domain detection due to the change of the sampling interval. Note that the smaller sampling interval means the finer detection granularity of local extrema ratios, which increases the detection accuracy. 
{Moreover, in Figure \ref{fig_detectacc}(a),} our experiments show that \defenseName~still achieves around 85\% detection accuracy in the worst case where the sampling rate is 44.1 kHz. We note that 44.1 kHz is the minimum sampling rate of common electronic devices in our lives~\cite{44100Hz}.  {Therefore, \defenseName~can achieve a good detection accuracy with different sampling rates in common devices.} 

\begin{figure}[t]
    \vspace{-0.1in}
    \centering
    \subfloat[accuracy vs. sampling rate.]{\includegraphics[width=1.66in]{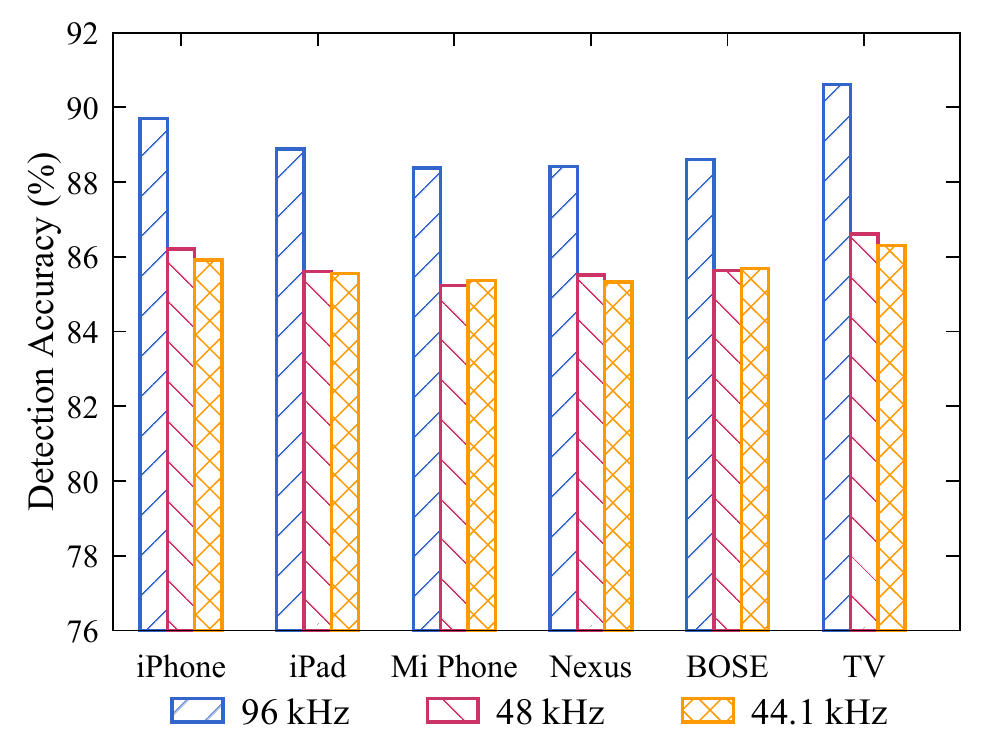}%
    \label{fig_sampling}}
    \subfloat[accuracy vs. noise level.]{\includegraphics[width=1.66in]{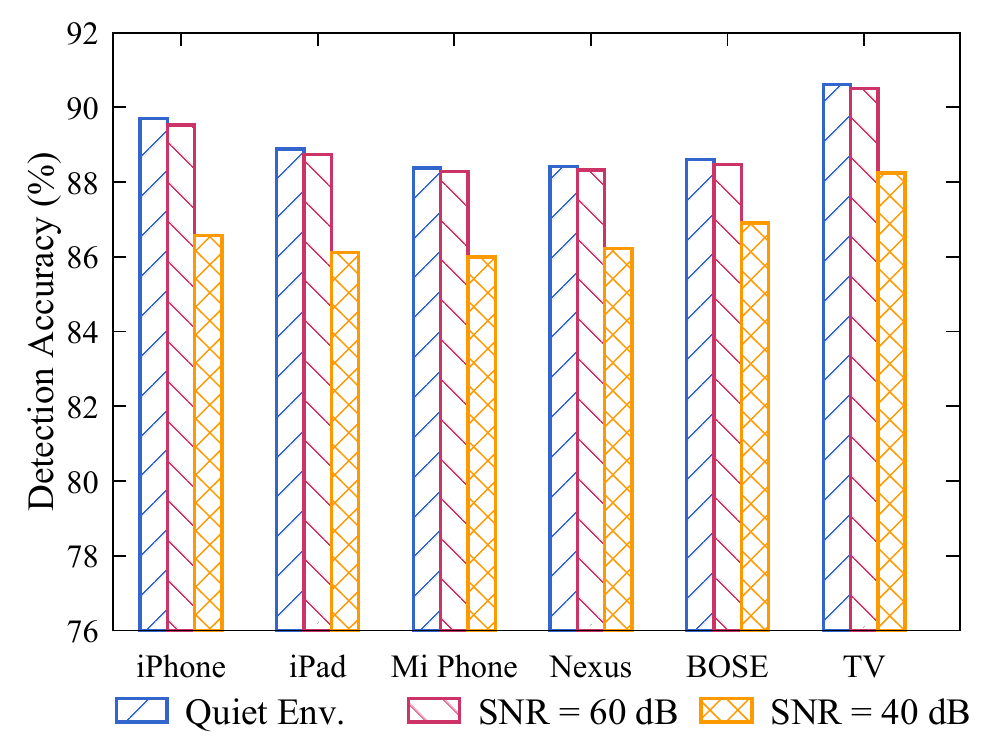}%
    \label{fig_noise}}
    \caption{Detection accuracy of different recording devices with different factors.}
    \label{fig_detectacc}
    \vspace{-0.15in}
\end{figure}

  


\vspace{0.03in}
\noindent \textbf{Impact from Different Recording Devices.}  {In Figure~\ref{fig_detectacc}(a), the detection accuracy does not significantly change when we use different recorders with the same sampling rate. The detection accuracy changes less than 2\% with different recording devices when the sampling rate is 48 kHz or 44.1 kHz. The results show \defenseName~can be applied to different recording devices in our lives.}


\vspace{0.03in}
\noindent \textbf{Impact from Different Noisy Environments.} To test the detection accuracy under different noisy environments, we introduce noise factors in our experiments. We test our detection method under three scenarios: (1) in a quiet environment, (2) in a noisy environment with the signal-to-noise ratio (SNR) of 60 dB, and (3) in a noisy environment with the SNR of 40 dB.
The additive noise signal is produced by a loudspeaker that plays a  pre-prepared Gaussian white noise signal, simulating the noise in the real world. The noise is mixed with the test signals with specific SNR. 
Figure~\ref{fig_detectacc}(b) illustrates the detection accuracy in various noise conditions.  {We can see that the impact of noise is limited. Particularly}, the detection accuracy remains unchanged when the SNR is 60 dB. When the SNR drops to 40 dB, the detection accuracy decreases by 3.2\% on average. Actually, the impact of noise is mainly reflected in the time-domain defense. General noise has little effect on the frequency-domain defense part. With the increase of noise power, the burr amplitude in the noise will also increase. As a result, noise can result in the imprecise detection of the local extrema pattern in the test signals. However, our experimental results indicate that \defenseName{} still works well at the general ambient noise level.

\vspace{-0.05in}
\subsection{ {Overhead of Dual-Domain Detection}}
 
We implement \defenseName~with C++ language, and build a system prototype in ReSpeaker Core v2, which is a popular voice interactive platform with quad-core ARM Cortex-A7 of 1.5GHz and 1GB RAM on-board.  {Our experimental results show that} the embedded program takes 5.5 $ms$ on average to process a signal segment of 32 $ms$ length with the CPU usage of 24.2\%. The largest memory usage of the program is 12.05 MB.  {The results demonstrate} the feasibility of applying our dual-domain detection system in the real world.


\vspace{-0.05in}
\section{Related Work}\label{sec:related}
In this section, we review related research on attacks targeting ASR systems, techniques on loudspeaker frequency response compensation, and defense systems against replay attacks, respectively.
 
\vspace{0.03in}
\noindent \textbf{Attacks on Speaker Dependent ASRs.} A speaker dependent ASR system is designed to only accept voice commands from specific users~\cite{SpoofingASR}. It verifies the speaker's identity by matching the individual characteristics of human voice. There are four main spoofing attacks against the speaker dependent ASRs. First, an attacker can physically approach a victim's system and alter its voice to impersonate the victim~\cite{Impersonation}. Second, the attacker can launch a simple replay attack by playing back a pre-recorded speech of the victim to the ASR systems~\cite{REPLAY2, REPLAY1}. Third, speech synthesis attacks generate artificial speech to spoof the ASR systems~\cite{SyntheticSpeech, Synthetic, AdversarySamples}. Fourth, speech conversion attacks aim to achieve a speech-to-speech conversion, so that the generated speech has the same timbre and prosody with the victim speech~\cite{VConversion1, VConversion2}.

\vspace{0.03in}
\noindent \textbf{Attacks on Speaker Independent ASRs.} A speaker independent system is designed to accept commands from any person without identity verification. Comparing to the speaker dependent system, it is more vulnerable to attacks~\cite{GVSAttack,A11Y,Monkey,AuDroid}. Recently, researchers found more surreptitious attacks that humans cannot easily perceive or interpret. Dolphin attack is hard to be noticed since the malicious audio is modulated into the ultrasonic range~\cite{Dolphin, inaudible_0, InaudibleVC}. The voice commands can also be modulated into laser light to launch audio injection attack~\cite{LightCommands}. Also, the malicious audio can be perturbed into an unintelligible form in either time domain or frequency domain~\cite{SPHidden}. To attack the machine learning module in ASRs, recent research shows attackers can produce noise-like~\cite{HidCmd, HiddenVCS, CocaineNoodles, CCS19POSTER} or song-like~\cite{CmdSong} voice commands that cannot be interpreted by human. Psychoacoustic model can also be applied to generate the adversarial audio below the human perception threshold~\cite{Psychoacoustic}. By fooling the natural language processing (NLP) module after ASRs, skill squatting attacks mislead the system to launch malicious applications~\cite{SkillSquatting, DangerousSkills, MIM, Paralinguistics, AfterASR}.

\noindent \textbf{Loudspeaker Frequency Response Compensation.} In the field of room acoustics, loudspeaker frequency response compensation is a technique used to improve the sound reproduction~\cite{Equalization0}. The basic method is to design an intelligent filter to flatten the frequency response of the loudspeakers~\cite{Equalization1}. The frequency response compensation can also be achieved by advanced filter with a generic Hammerstein loudspeaker model~\cite{Equalization2}. For a multichannel loudspeaker system, the minimax approximation method is proposed to flatten the spectral response around the crossover frequency~\cite{Equalization3}. Also, a polynomial based MIMO formulation is proposed to solve the multi-speaker compensation problem~\cite{Equalization4}. 

\vspace{0.03in}
\noindent \textbf{Defenses against Replay Attacks.} In ASVspoofing Challenge~\cite{ASV2017}, several replay detection methods are proposed by exploiting the frequency-based features, such as Linear Prediction Cepstral Coefficient (LPCC)~\cite{LPCC}, Mel Frequency Cepstral Coefficient (MFCC)~\cite{MFCC}, Constant Q Cepstral Coefficients (CQCC)~\cite{CQCC}, High Frequency Cepstral Coefficients (HFCC)~\cite{HFCC} and Modified Group Delay Cepstral Coefficient (MGDCC)~\cite{MGDCC}. 
Besides, the high-frequency sub-band features can be used to detect live human voice by the linear prediction (LP) analysis~\cite{freq_solution}. The sub-bass (low-frequency range) energy is also an effective feature to detect the replay signals, though this method can be bypassed by altering the speaker enclosure or modulating the signals with our inverse filter~\cite{subbass}. The frequency modulation features~\cite{FM1, FM2, ModDynamic} can also be leveraged due to the degraded amplitude components of replay noise.

Researchers propose to detect replay attacks using physical properties. Gong et al. detect the body-surface vibration via a wearable device to guarantee the voice comes from a real user~\cite{AcousticCues}. 2MA~\cite{2MA} verifies the voice commands by sound localization using two microphones. Yan et al. propose a spoofing detection method based on the voiceprint difference between the authentic user and loudspeakers~\cite{CCS19Fieldprint}. All these methods require special equipment or specific scenarios. VoiceLive~\cite{VoiceLive} detects live human voice by capturing the time-difference-of-arrival (TDoA) dynamic of phoneme sound locations. VoiceGesture~\cite{Liveness} reuses smartphones as a Doppler radar and verifies the voice by capturing the articulatory gesture of the user when speaking a passphrase. However, these two methods work well only when there is a short distance between the recorder and the user’s mouth. 
\vspace{-0.05in}
\section{Conclusion}\label{sec:conclusion}

In this paper, we propose a new modulated replay attack against ASR systems. This attack can bypass all the existing replay detection methods that utilize different frequency domain features between electronic speakers and humans. We design an inverse filter to help compensate frequency distortion so that the modulated replay signals have almost the same frequency features as human voices. To defeat this new attack, we propose a dual-domain defense that checks audio signal's features in both frequency domain and time domain. Experiments show our defense can effectively defeat the modulated replay attacks and classical replay attacks.


\vspace{-0.05in}
\section*{Acknowledgments}
This work is partially supported by the U.S. ARO grant W911NF-17-1-0447, U.S. ONR grants N00014-18-2893 and N00014-16-1-3214, and the NSFC grants U1736209 and 61572278. Jiahao Cao and Qi Li are the corresponding authors of this paper. 
\bibliographystyle{ACM-Reference-Format}
\bibliography{paperref}
\begin{appendices}  
\section{Mathematical Proof of Ringing Artifacts in Modulated Replay Audio}\label{append:proof}
\begin{theorem} \label{theory:1}

Uncertainty Principle: It is hard to accurately determine the entire frequency response of a loudspeaker.

\end{theorem}
\begin{proof}

The frequency response of a loudspeaker contains amplitude response and phase response. The measurement of amplitude response is demonstrated in Section~\ref{sec:if}. However, it is difficult to accurately measure the phase response.

For an electronic circuit system, the phase response can be measured by observing the electric signals $x_{out}(t)$ and $x_{in}(t)$ with an oscilloscope. But in a loudspeaker system, we cannot measure the phase response directly because the output signal $x_{out}(t)$ is a sound wave. Other equipment (such as a receiver that converts sound wave to electronic signal) is required to complete the measurement. But the measuring system can introduce other phase differences. There are mainly three influence factors:

(1) Time of flight. The propagation time will add phase differences. It is important to know the accurate delay time $t = L/v_{0}$, where $L$ is the direct distance between the speaker and the sensor. The sound speed $v_{0} \approx 344 $m/s (@20$^{\circ}$C).

(2) \emph{Time incoherence}. Most of the available loudspeakers are not time coherent, which will exhibit phase error in the measurement.

(3) \emph{Phase response of receiving sensor}. The phase response of receiving sensor is typically unknown, which will also introduce phase shifts.

As a result, the accuracy of phase response measurement cannot be guaranteed. That means the entire frequency response cannot be accurate. Also we can prove that even small measurement errors for phase response can cause ringing artifacts (see Theorem 3).

\end{proof}

\begin{theorem} \label{theory:2}

Compared to the genuine signal $x(t)$, there are phase shifts for each frequency component in the modulated replay signal $x_{mr}(t)$.

\end{theorem}
\begin{proof}

In the modulated replay attack, the inverse filter only needs to compensate the amplitude spectrum because the features (e.g. CQCC, MFCC, LPCC) in the existing defenses only derives from the amplitude spectrum. However, a loudspeaker has a non-zero phase response in the real world, though it cannot be accurately measured (see Theorem~\ref{theory:1}).

Suppose the genuine audio x(t) is a digital signal. Through the fast Fourier transform, x(t) would be decomposed as $N$ frequency components with the frequency set $\{ f_{1}, f_{2},...,f_{N}\}$. The frequency spectrum of $x(t)$ is denoted as $\{ A_{n}, \varphi_{n} \}$, where $\{ A_{n} \}$ is the amplitude spectrum while $\{ \varphi_{n} \}$ is the phase spectrum. So, $x(t)$ can be represented as 

\begin{equation}
x(t) = \sum_{n} A_{n} \cdot \sin(2 \pi f_{n} t + \varphi_{n}).
\end{equation}

Assume that the frequency response of the loudspeaker is $H = \{ G_{n}, \psi_{n} \}$, where $\{ G_{n} \}$ is the amplitude response while $\{ \psi_{n} \}$ is the phase response. By measuring the input and output test signals, attacker can achieve the estimated frequency response $\hat{H} = \{ \hat{G_{n}}, 0 \}$. 

The inverse filter is then designed based on $\hat{H}$, denoted as $I = \hat{H}^{-1} = \{ \hat{G_{n}}^{-1}, 0 \}$. As a result, the generated modulated audio would be 

\begin{equation}
x_{m}(t) = \sum_{n} (A_{n}/\hat{G_{n}}) \cdot \sin(2 \pi f_{n} t + \varphi_{n}).
\end{equation}

If the loudspeaker is ideal that does not have phase shift effects. And the amplitude estimation is enough accurate. The estimated replay output of the modulated audio would be 

\begin{equation}
\begin{aligned}
\hat{x_{mr}}(t) & = \sum_{n} (A_{n} \cdot G_{n} /\hat{G_{n}}) \cdot \sin(2 \pi f_{n} t + \varphi_{n})\\
& \approx \sum_{n} A_{n} \cdot \sin(2 \pi f_{n} t + \varphi_{n}) = x(t),
\end{aligned}
\end{equation}

\noindent which is approximately equal to the genuine audio.

However, if the modulated audio $x_{m}(t)$ passes through the real loudspeaker system $H$, the real modulated replay $x_{mr}(t)$ audio would be

\begin{equation}
\begin{aligned}
x_{mr}(t) & = \sum_{n} (A_{n} \cdot G_{n} /\hat{G_{n}}) \cdot \sin(2 \pi f_{n} t + \varphi_{n} + \psi_{n})\\
& \approx \sum_{n} A_{n} \cdot \sin(2 \pi f_{n} t + \varphi_{n} + \psi_{n}) \neq x(t).
\end{aligned}
\end{equation}

Because $x_{mr}(t)$ has almost the same amplitude spectrum with the genuine audio $x(t)$, it can bypass the existing defense systems. However, compared to the genuine signal $x(t)$, there are phase shifts for each frequency component in the modulated replay signal $x_{mr}(t)$.

\end{proof}

\begin{theorem} \label{theory:3}

The phase shifts will cause the spurious oscillations (ringing artifacts) in the original audio.

\end{theorem}
\begin{proof}

\newcommand*{\dif}{\mathop{}\!\mathrm{d}}

Suppose there is a small phase shift $\dif \varphi$ in the $N$-th frequency component of the signal $x(t)$, while other frequency components remain unchanged. The new signal would be 

\begin{equation}
\begin{aligned}
x'(t) &= \sum_{n \neq N} A_{n} \cdot \sin(2 \pi f_{n} + \varphi_{n}) + A_{N} \cdot \sin(2 \pi f_{N} + \varphi_{N} + \dif \varphi)\\
& = \sum_{n} A_{n} \cdot \sin(2 \pi f_{n} + \varphi_{n}) + A_{N} \cdot \sin(2 \pi f_{N} + \varphi_{N} + \dif \varphi)\\
& \quad - A_{N} \cdot \sin(2 \pi f_{N} + \varphi_{N})\\
& = x(t) + 2 \cdot A_{N} \cdot \sin(\frac{\dif \varphi}{2}) \cdot \cos(2 \pi f_{N} + \varphi_{N} + \frac{\dif \varphi}{2})\\
& = x(t) + C \cdot \cos(2 \pi f_{N} + \varphi_{N} + \frac{\dif \varphi}{2})\\
& = x(t) + o_{N}(t).
\end{aligned}
\end{equation}

Because $\dif \varphi$ is a very small shift value, $C$ is a small constant that satisfies $\left| C \right| < \left| A_{n} \cdot \dif \varphi \right|$. 

$x(t)$ is an audio signal that is statistically smooth in the time domain. Hence, the new signal $x'(t)$ contains small ringing artifacts because of the additional oscillations signal $o_{N}(t)$ with the frequency of $f_{N}$. The maximum amplitude of the spurious oscillations is limited by $\left| C \right|$ value.

Assume that the phase shifts of a loudspeaker system are denoted as $\psi = \{ \psi_{n} \}$ for all frequency components. The modulated replay signal would be 

\begin{equation}
\begin{aligned}
x_{mr}(t) & = \sum_{n} A_{n} \cdot \sin(2 \pi f_{n} + \varphi_{n} + \psi_{n})\\
& = x(t) + 2 \cdot \sum_{n} A_{n} \cdot \sin(\frac{\psi_{n}}{2}) \cdot \cos(2 \pi f_{n} + \varphi_{n} + \frac{\psi_{n}}{2})\\
& = x(t) + \sum_{n} C_{n} \cdot \cos(2 \pi f_{n} + \varphi_{n} + \frac{\psi_{n}}{2})\\
& = x(t) + o(t).
\end{aligned}
\end{equation}

The total spurious oscillations $o(t)$ can be presented as

\begin{equation}
o(t) = 2 \cdot \sum_{n} A_{n} \cdot \sin(\frac{\psi_{n}}{2}) \cdot \cos(2 \pi f_{n} + \varphi_{n} + \frac{\psi_{n}}{2}).
\end{equation}

The maximum amplitude $A_{o}$ of the spurious oscillations is constraint by the following condition.

\begin{equation}
A_{o} = \sum_{n} \left| C_{n} \right| < \sum_{n} A_{n} \cdot \left| \psi_{i} \right|
\label{eq:constraint}
\end{equation}

As a result, the phase shifts of the loudspeakers will lead to the ringing artifacts in the modulated replay audio.

\end{proof}

\end{appendices}
\begin{appendices}
\section{Parameters in Detection Methods}\label{append:par}

We list the parameters of different replay detection methods here for better understanding the modulated replay attack.

\vspace{0.05in}
\noindent {\bf (1) Constant Q Cepstral Coefficients (CQCC) based method. } 
The Constant-Q Transform (CQT) is applied with a maximum frequency of $F_{max}=f_{s}/2=48kHz$. The minimum frequency is set to $F_{min}=F_{max}/2^{12}=11.7Hz$ (12 is the number of octaves). The value of bins per octave is set to 96. Re-sampling is applied with a sampling period of $d=16$. The dimension of the CQCC features is 19. Experiments were performed with all possible combinations of static and dynamic coefficients.


\vspace{0.05in}
\noindent {\bf(2) Mel Frequency Cepstral Coefficents (MFCC) based method. }
The window length is set to 3072 samples (32 ms), and the window shift is 1536 samples (16 ms). Thus, the frequency bins would be 4096 samples. When we create the triangular mel-scale filterbanks, the number of filterbanks is 26. The length of each filter is set to 2049. The sampling rate in experiments is 96 kHz.


\vspace{0.05in}
\noindent {\bf(3) Linear Predictive Cepstral Coefficients (LPCC) based method. }
In the LPCC feature, the frame length is set to 1280 and the offset is 0. The threshold of the silence power is $10^{-4}$. The prediction order in the LPC coefficients is set to 14. 

\vspace{0.05in}
\noindent {\bf(4) Mel Wavelet Packet Coefficients (MWPC) based method. }
MWPC feature is based on wavelet packet transform, adapted to the mel scale. Instead of using the energy of the frequency sub-bands, MWPC use Teager Keiser Energy (TKE) Operator as the following equation, $\Psi(s(t)) = s(t)^2 - s(t-1)s(t+1)$. The dimension of MWPC features is 12, derived from the principle component analysis.  

\vspace{0.05in}
\noindent {\bf(5) High-frequency sub-band power based method. } High frequency energy ratio is measured between (2-4) kHz and (0-2) kHz.

\vspace{0.05in}
\noindent {\bf(6) High-frequency CQCC based method. }
Similar to CQCC-based methods. But it concerns the high-frequency (2-4kHz) band.

\vspace{0.05in}
\noindent {\bf(7) FM-AM based method. } 
This method aim to detect the frequency modulation (FM) and amplitude modulation (AM) features in replay audio. Here, the feature vector consists of the modulation centroid frequency (MCF) and modulation static energy (MSE). Which are both extracted from modulation spectrum. The Gaussian mixture model (GMM) is employed as the back-end classifier.

\vspace{0.05in}
\noindent {\bf(8) Sub-bass Frequency based method. }
Energy balance metric indicates the energy ratio of the sub-bass range (20-80 Hz) to the low-frequency range (20-250 Hz). The threshold is set to 0.228 according to the study~\cite{subbass}.

\end{appendices}
\begin{appendices}
\section{Inverse Filter Implementation}\label{append:spk}

The speaker response estimation process contains two steps: discrete amplitude response measurement and continuous amplitude response fitting. 
In the discrete amplitude response measurement, we measure the speaker input/output response coefficient by testing 68 discrete typical frequency values. The discrete frequency values are within four audio frequency ranges: bass (from 60 Hz to 225 Hz with a spacing of 15 Hz), low midrange (from 250 Hz to 500 Hz with a spacing of 50 Hz), midrange (from 550 Hz to 2 kHz with a spacing of 50 Hz), and upper midrange (from 2.1 kHz to 4 kHz with a spacing of 100 Hz). {The input test signals are single-frequency signals with the same amplitude of 1, which are generated by using the \emph{wavwrite} tool and stored in a lossless format. The test audio is then transferred to replay devices and played at medium volume on loudspeakers, since the response function is not directly related to the input amplitude according to our experiments.} After the spectrum analysis, we can get a rough response polygonal curve across 68 discrete points. 

In the finer-grained amplitude response fitting, we need to first calculate the spectral resolution of the modulated signal $\Delta f = f_{s} / N$, where $f_{s}$ is the signal sampling rate. $N$ is the FFT point number which is the minimum power of 2 that is greater than or equal to the signal length $L$, denoted as $N = 2 ^ { \lceil \log_{2} L \rceil }$. The finer-grained amplitude response curve can be achieved by the cubic spline fitting. And the estimated response used in the inverse filter generation is sampled with the signal spectral resolution $\Delta f$. The inverse filter is designed by using the finer-grained speaker response $H(k)$. In order to avoid divide-by-zero error in our experiments, the inverse filter transfer function is calculated as $1/(H(k)+eps)$, where $eps$ is a small value from 0.001 to 0.002.

\begin{figure}[t]
\vspace{-0.1in}
\centering
    \subfloat[iPhone X]{
        \includegraphics[width=1.6in]{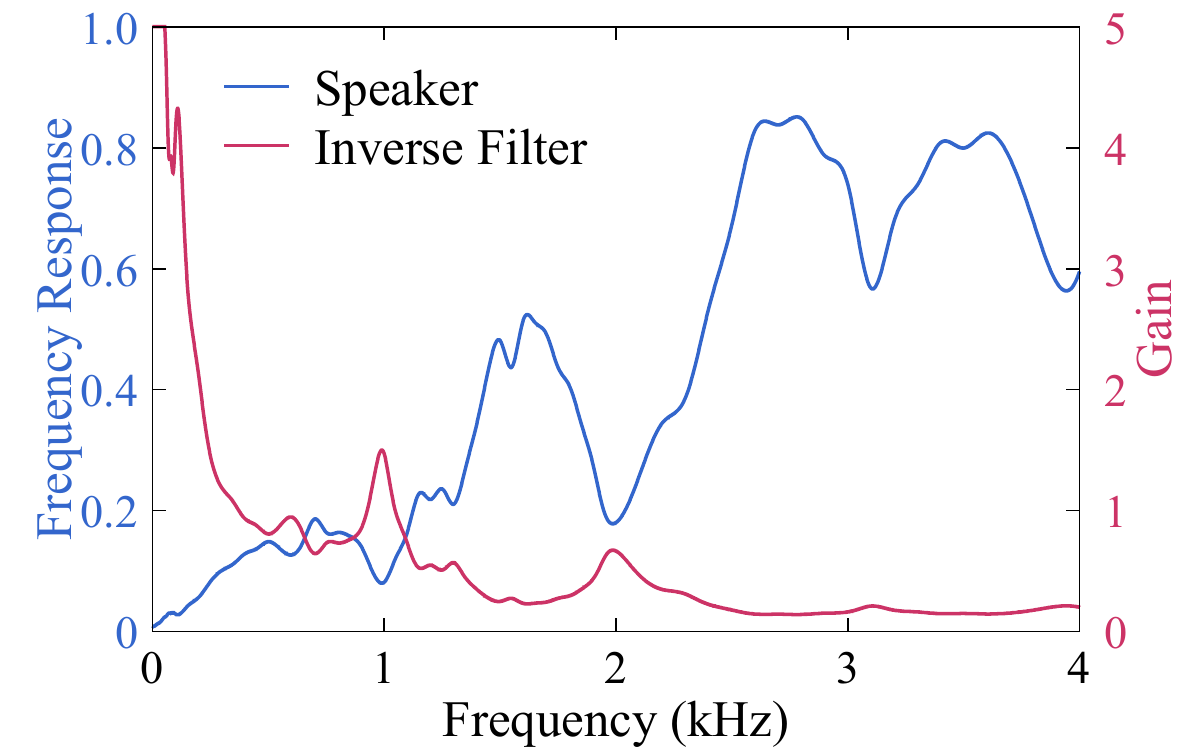}
        }%
    \subfloat[iPad Pro]{
        \includegraphics[width=1.6in]{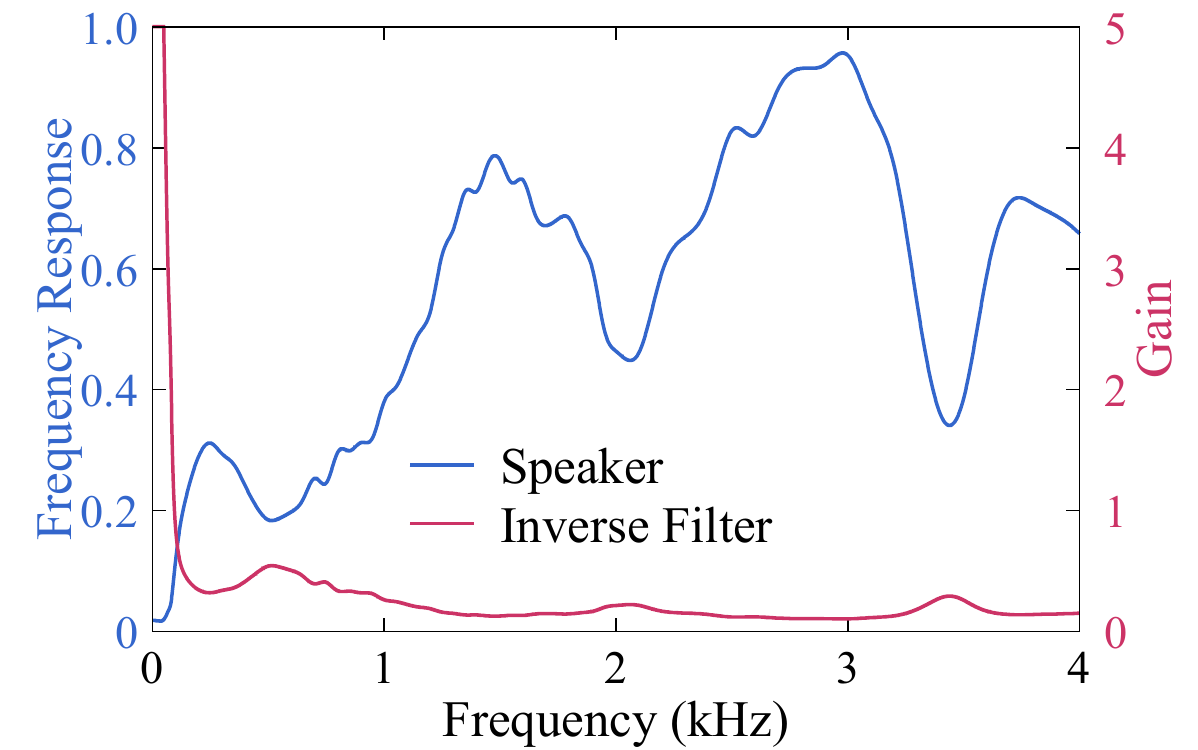}
        }%
    \hfil
    \subfloat[Mi Phone 4]{
        \includegraphics[width=1.6in]{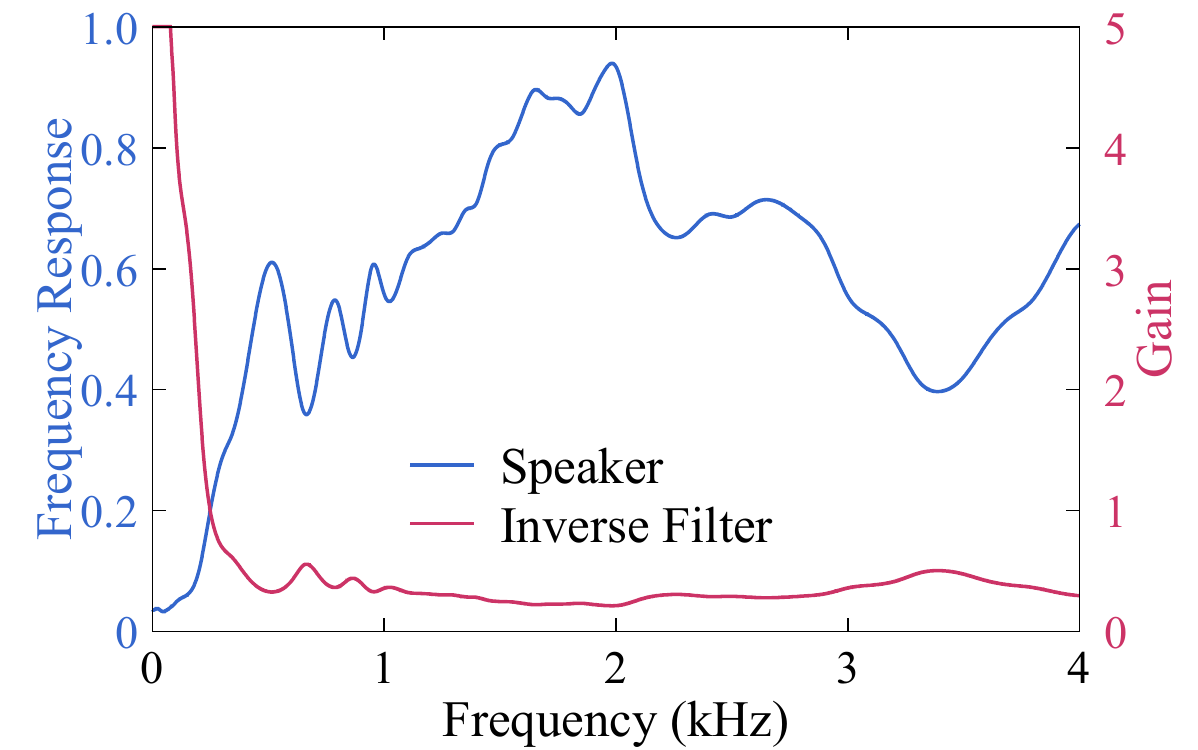}
        }%
    \subfloat[Google Nexus 5]{
        \includegraphics[width=1.6in]{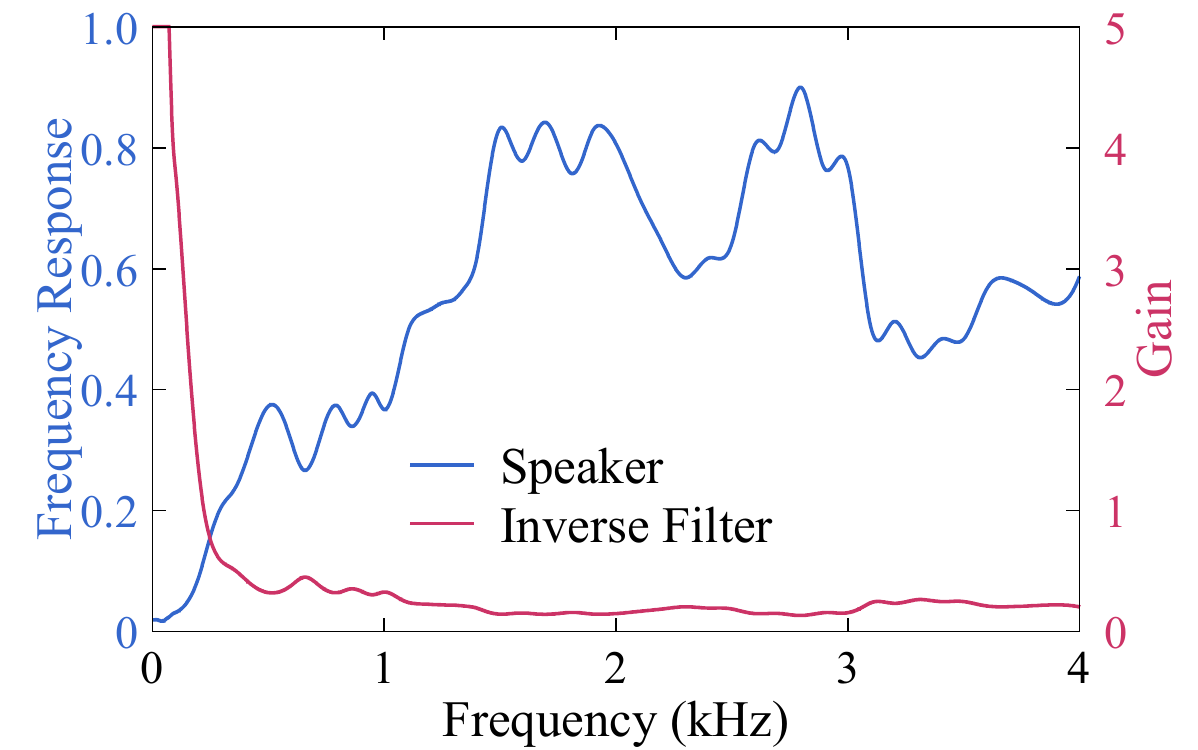}
        }%
    \hfil
    \subfloat[Bose Soundlink Micro]{
        \includegraphics[width=1.6in]{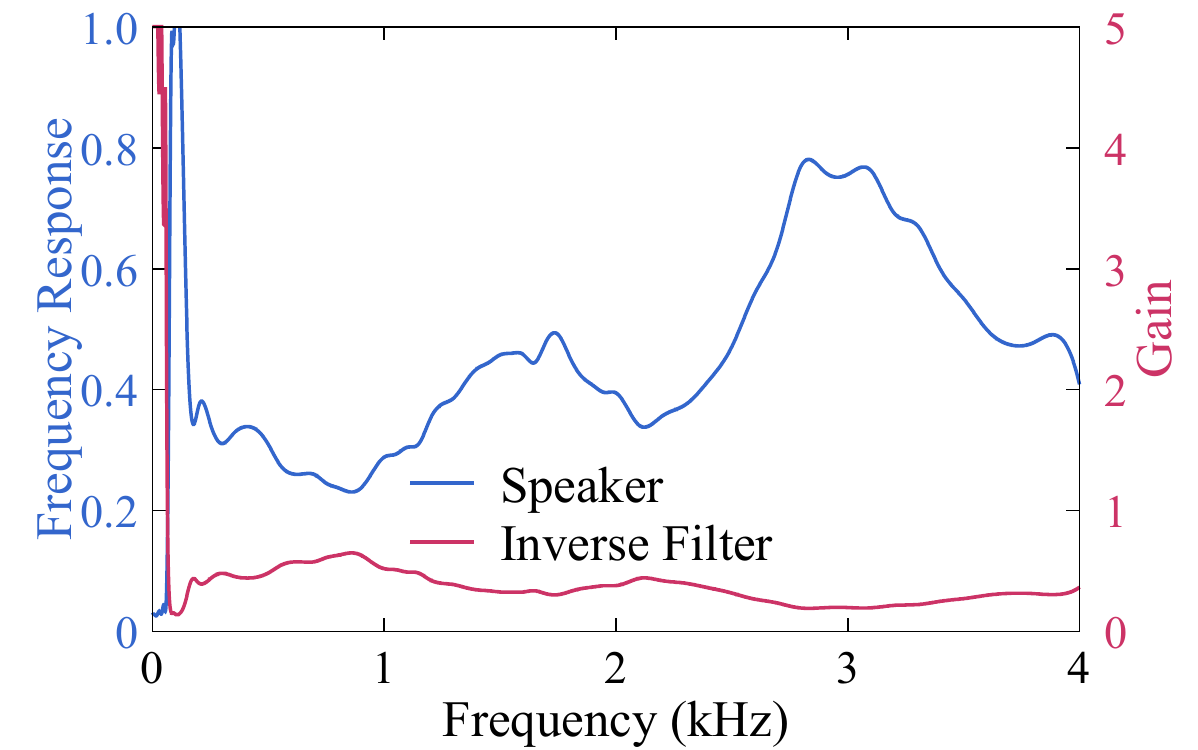}
        }%
    \subfloat[Samsung Smart TV]{
        \includegraphics[width=1.6in]{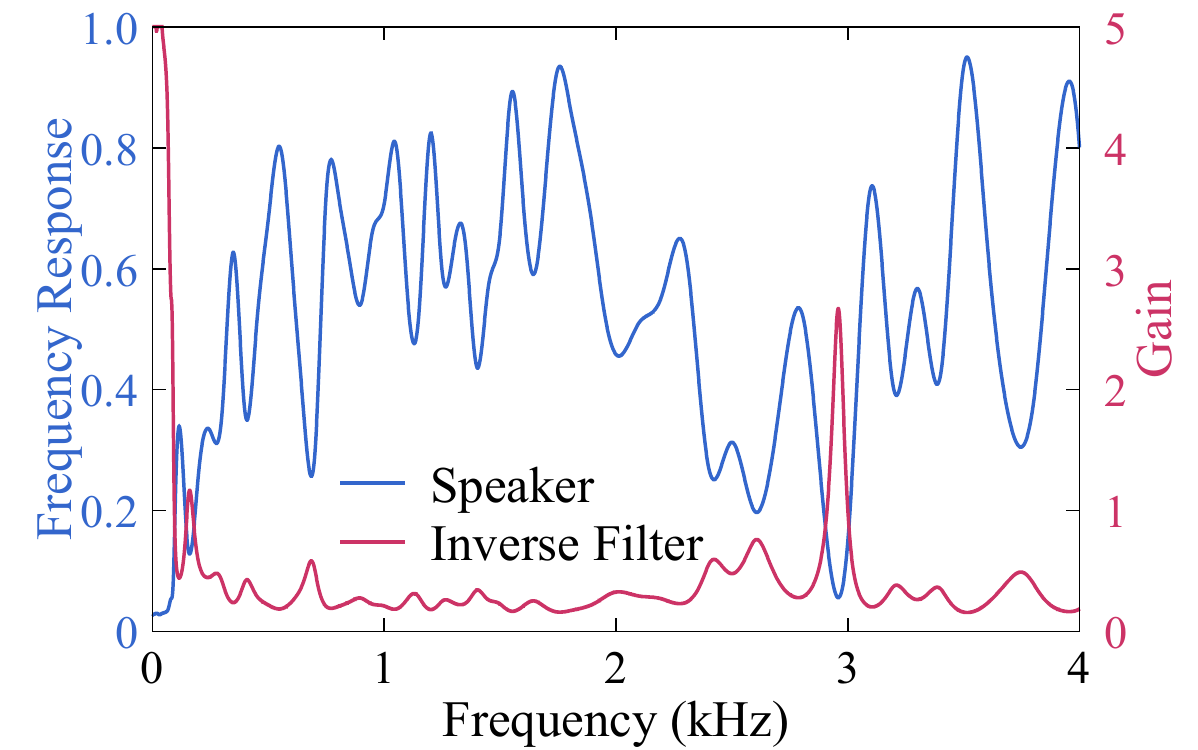}
        }
    \centering
    \caption{The amplitude response curves of different speaker devices and their corresponding inverse filters.}
    \label{fig_exp0}
\vspace{-0.15in}
\end{figure}

Figure~\ref{fig_exp0} shows the amplitude response curves of different speaker devices and their inverse filters. For mobile devices, the response curves are high-pass filters due to the limited size of speakers. Therefore, the inverse filters should be low-pass filters. For Bose Soundlink Micro which has a tweeter and a woofer, there are obvious two-stage enhancements in the  amplitude response. However, the transfer function still cannot be considered as a pass-through filter. The frequency response of Samsung Smart TV fluctuates with frequency due to its two speakers  that create stereo audio. We can use designed inverse filters to compensate the speaker amplitude response, mitigating the decay of frequency components.

\end{appendices}
\begin{appendices}
\section{Classifiers in Time-Domain Defense}\label{append:classifier}

In the time-domain defense, the local extreme ratio (LER) is a robust feature that can describe the ringing artifacts in modulated replay audios. Therefore, the classifier selection has little impact on the defense  performance. To verify this hypothesis, we conduct experiments to evaluate the effects of different classifiers on the feature classification.

We classify the LER features using five common classifiers, including Support Vector Machine (SVM), Decision Tree (DT), Naive Bayes (NB), Gaussian Mixture Model (GMM), and K-Star. The 10-fold cross-validation accuracy is used as the evaluation standard. The performance of different classifiers is shown in Figure~\ref{fig_classselect}. We can see that SVM, Decision Tree, and KStar achieve better performance than other classifiers. Gaussian Mixture Model obtains the worst accuracy since the data distribution of LER features does not subject to the normal distribution. Above all, we choose the SVM model in our system due to its easy deployment and high performance.

\begin{figure}[t]
  \centering
  \includegraphics[width=2.5in]{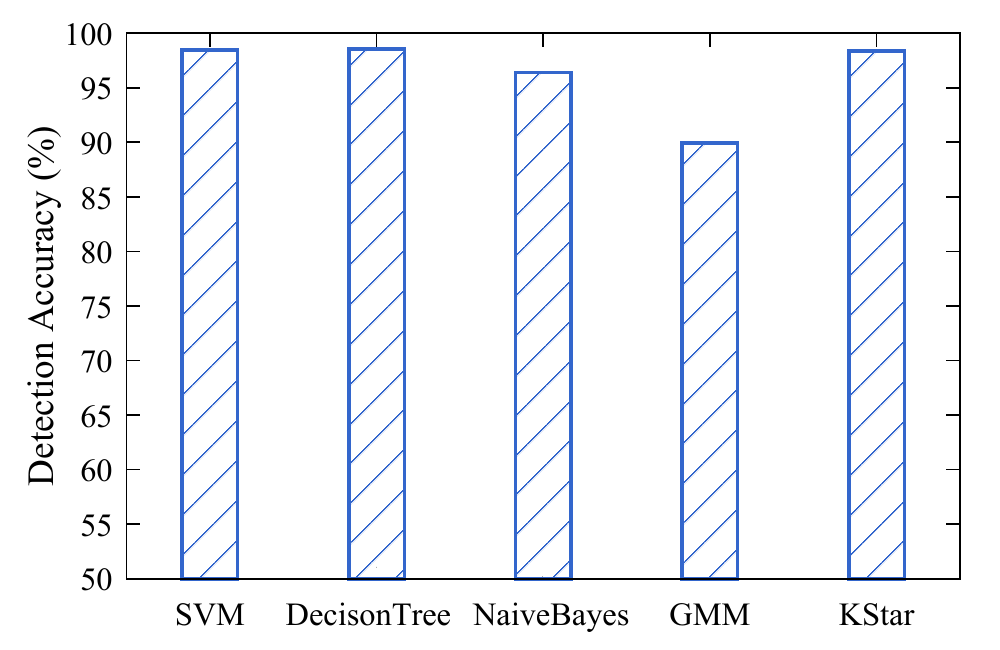}
  \vspace{-0.1in}
  \caption{Performance of different classifiers in the time-domain defense.}
  \vspace{-0.1in}
  \label{fig_classselect}
\end{figure}

\end{appendices}

\end{document}